\documentclass[11pt]{article}
\usepackage[noadjust]{cite}
\usepackage{tikz}
\usepackage{longtable}
\usepackage{caption}
\usepackage{multirow,multicol}
\usepackage{epsf, graphicx}
\usepackage{latexsym,amsfonts,amsbsy,amssymb}
\usepackage{amsmath,amsthm}
\usepackage{cite}
\usepackage{cleveref}
\usepackage{indentfirst}
\usepackage{bm}
\usepackage [latin1]{inputenc}
\usepackage{enumitem}
\usepackage{lipsum}
\usepackage{tabularx}
\usepackage{adjustbox}
\setenumerate[1]{itemsep=0pt,partopsep=0pt,parsep=\parskip,topsep=5pt}
\setitemize[1]{itemsep=0pt,partopsep=0pt,parsep=\parskip,topsep=5pt}
\setdescription{itemsep=0pt,partopsep=0pt,parsep=\parskip,topsep=5pt}
\allowdisplaybreaks[4]
\makeatletter

\@addtoreset{equation}{section} \makeatother
\newtheorem{Theorem}{Theorem}[section]
\newtheorem{Lemma}{Lemma}[section]
\newtheorem{Corollary}{Corollary}[section]
\newtheorem{Remark}{Remark}[section]
\newtheorem{Definition}{Definition}[section]
\newtheorem{Proposition}{Proposition}[section]

\newtheorem{Example}{Example}[section]
\makeatletter \@addtoreset{equation}{section} \makeatother
\begin{document}
	
	\title{Function-Correcting Codes with Homogeneous Distance \let\thefootnote\relax\footnotetext{E-Mail addresses: huiyingliu@mails.ccnu.edu.cn (H.Y. Liu),  hwliu@ccnu.edu.cn (H. Liu).}}
	\author{ Huiying Liu,~Hongwei Liu }
	\date{\small School of Mathematics and Statistics, Central China Normal University, Wuhan, 430079, China}
	\maketitle
	%MS+++++++++++++++++++++ Abstract +++++++++++++++++++++++++
	{\noindent\small{\bf Abstract:} Function-correcting codes are designed to reduce redundancy of codes when protecting function values of information against errors. As generalizations of Hamming weights and Lee weights over $ \mathbb{Z}_{4} $, homogeneous weights are used in codes over finite rings. In this paper, we introduce function-correcting codes with homogeneous distance denoted by FCCHDs, which extend function-correcting codes with Hamming distance. We first define $ D $-homogeneous distance codes. We use $ D $-homogenous distance codes to characterize connections between the optimal redundancy of FCCHDs and lengths of these codes for some matrices $ D $. By these connections, we obtain several bounds of the optimal redundancy of FCCHDs for some functions. In addition, we also construct FCCHDs for homogeneous weight functions and homogeneous weight distribution functions. Specially, redundancies of some codes we construct in this paper reach the optimal redundancy bounds.}
	
	\vspace{1ex}
	
	\section{Introduction}
    Function-correcting codes were first introduced in coding theory by Lenz, Bitar, Wachter-Zeh and Yaakobi \cite{lbwy}, which can be viewed as generalizations of unequal error protection codes denoted by UEP codes. UEP codes were proposed by Masnick and Wolf \cite{mw}. Further analysis of UEP codes can be found in \cite{bnz}, \cite{bk}. UEP are mainly divided into bit-wise UEP and message-wise UEP. And the work of message-wise UEP (denoted by UMP) codes can be found in \cite{ssgagd}, where UMP codes provides unequal message protection. The authors considered function-correcting codes with Hamming distance denoted by FCCs in \cite{lbwy}. The goal of the study of FCCs is to reduce the redundancy compared with traditional codes. The reason why FCCs can reduce the redundancy is that the receiver may not need all correct messages but only the attribute of messages. Therein, the attribute of messages is a function related to messages. Obviously, if we know correct messages, we obtain the attribute of messages, thus some classical codes can be viewed as FCCs for identity function. In FCCs, Hamming distances between codewords satisfy appropriate requirements, where messages of these codewords have different function values. In fact, FCCs use system coding method, a codeword in FCCs consists of a message vector and a redundancy vector. Hamming distances between message vectors may be different, thus the redundancy vectors of an FCC have irregular distance requirements. Based on this fact, the authors introduced a class of irregular distance codes to study FCCs. Using these codes, they obtained some bounds of the optimal redundancy of FCCs for some functions. In 2024, Premlal and Rajan \cite{pr} further studied FCCs for linear functions.

    With the development of information, we face the problem of storing a large amount of data. Cassuto and Blaum \cite{cb} first proposed symbol-pair read channels for processing high-density data storage. In 2024, Xia, Liu and Chen \cite{xlc} introduced FCCs with symbol-pair weights in symbol-pair read channels and studied the optimal redundancy of these codes.
    In 2025, Ge, Xu, Zhang and Zhang \cite{gxzz} improved some results in references \cite{lbwy} and \cite{xlc}. Singh, Singh and Yaakobi \cite{ssy} extended results of reference \cite{xlc} to $b$-symbol read channels.

    For some specific functions, the characterization of function-correcting codes may be more specific. There are some related works of function-correcting codes for specific funtions in previous study. Computing function values of combined datas from a sender and a receiver himself is related to the work of FCCs. Orlitsky and Roche \cite{or} studied this problem, they proposed an upper bound of the minimum number of bits. Kuzuoka and Watanable \cite{kw} introduced informative functions and studied the optimal rate for some specific functions by the structure of these functions. What's more, error correcting codes that protect the weights of robust neural networks stored in memory devices had been studied in \cite{hsj}. Huang, Siegel and Jiang \cite{hsj} proposed a better Selective Protection scheme to improve the performance of neural networks after correction.

    Homogeneous weights are generalizations of Lee weights on $ \mathbb{Z}_{4} $ and Hamming weights over finite fields. In 1997, Constantinescu
    and Heise \cite{hc} first introduced the concept of the homogeneous weight on the residue rings of integers. The equivalent characterization for homogeneous distance to be a metric can be seen in \cite{hc}, \cite{ll}. In 2000, Greferath and Schmidt \cite{gs} generalized the homogeneous weight to be defined on arbitrary finite rings. And the work about the existence and uniqueness of the homogeneous weight can be seen in \cite{gs}, \cite{nk}. Further work about homogeneous weights formula was studied in \cite{fl}, \cite{gs}, \cite{h}, \cite{vw}. Bounds of lengths of codes are significant for applications. In \cite{go}, \cite{ss}, bounds of codes with homogeneous distance were studied. In \cite{os}, \cite{sos}, the authors investigated the existence of perfect linear codes over $ \mathbb{Z}_{4} $ and $ \mathbb{Z}_{2^{l}} $ with the homogeneous weights, respectively.

    In this paper, we study function-correcting codes with homogeneous distance denoted by FCCHDs. FCCHDs use encoding functions with homogeneous distance to get the correct function values after channel transmission. Compared with traditional coding theory on finite ring, the goal of FCCHDs is to get correct function values of messages rather than correct messages. Thus this encoding method may reduce the redundancy compared with classical coding theory on finite rings.

    In this paper, we define irregular homogeneous distance codes which have irregular homogeneous distance requirements between codewords. Therein, these requirements are related to a given matrix $ D $ and we call this code as $ D $-homogeneous distance codes for the given matrix $ D $. Observing redundancy vectors of FCCHDs, we find some connections between lengths of these vectors and $ D $-homogeneous distance codes. To study the optimal redundancy of FCCHDs on the finite residue ring $ \mathbb{Z}_{2^{l}} $, we define two classes of matrices, which are homogeneous distance requirement matrices and function homogeneous distance matrices. Entries in these two classes of matrices are related to homogeneous distances between messages. Using these matrices to construct irregular homogeneous distance codes, we obtain some bounds of FCCHDs for an arbitrary function and some specific functions. In addition, for homogeneous weight functions and homogeneous weight distribution functions, we construct some FCCHDs for these functions. Therein, redundancies of some of these FCCHDs reach the optimal redundancy bounds of FCCHDs for these two classes of functions.

    This paper is organized as follows. Some definitions and formulas are given in Section 2. We present some bounds of the optimal redundancy of function-correcting codes with homogeneous distance for an arbitrary function in Section 3. In Section 4, we offer bounds of optimal redundancy of FCCHDs for homogeneous weight functions, homogeneous weight distribution functions and Rosenbloom-Tsfasnman weight functions. We give some specific FCCHDs for homogeneous weight functions and homogeneous weight distribution functions. And we find these FCCHDs reach the optimal redundancy bounds for these functions respectively. In Section 5, we study FCCHDs for locally homogeneous binary functions and min-max functions. We conclude this paper in the last section.

	\section{Preliminaries}
     Let $ |S| $ denote the size of a set $ S. $ $ k,r,l $ are positive integers and $ D $ is a matrix. In this section, we introduce the definition of homogeneous weight on a finite ring. We define function-correcting codes with homogeneous distance and a class of codes with irregular homogeneous distance requirements.

     \subsection{Homogeneous distance}
     In this subsection, we mainly introduce definitions of homogeneous weight and homogeneous distance on a finite ring.

    \begin{Definition}{\rm (\cite{gs})}\label{Def1}
    A real-valued function $\omega_{h}$ on a finite ring $R$ is called a (left) \textit{homogeneous weight} if $\omega_{h}(0)=0$ and the followings hold:\\
    {\rm (1)} For all $x,y \in R,~Rx=Ry$ implies $\omega_{h}(x)=\omega_{h}(y)$ holds.\\
    {\rm (2)} There exists a non-negative real number $\gamma$ such that
   $$ \sum_{y \in Rx}\omega_{h}(y)=\gamma|Rx| \text{ for all } 0\neq x \in R,$$
   where $ Rx $ denotes the left ideal of $ R $ generated by $ x. $
    \end{Definition}

    By {\rm \cite{gs}}, without loss of generality, let the parameter $ \gamma $ be 1. By {\rm \cite{fl,ll}}, we obtain a formula to calculate the homogeneous weight on the residue class ring $\mathbb{Z}_{2^{l}}, $ where $ l $ is a positive integer. For any $x \in  \mathbb{Z}_{2^{l}} , $
        \begin{equation}\label{Formula}
           \omega_{h}(x)=\left\{
  	         \begin{aligned}
  	          0&,~\text{if } x=0,\\
                1&,~\text{if } x \notin \langle 2^{l-1} \rangle,\\
                2&,~\text{if } x \in \langle 2^{l-1} \rangle \setminus \{0\}.
               \end{aligned}
               \right.
        \end{equation}

    \begin{Definition}
        Let $ R $ be a finite ring with homogeneous weight $ \omega_{h} $. Homogeneous distance between $ x,\,y\,\in R $ is defined as $ d_{h}(x,y)=\omega_{h}(x-y). $
    \end{Definition}
    Obviously, if $ u=(u_{1},u_{2},\ldots,u_{k}) \in \mathbb{Z}_{2^{l}}, \, \omega_{h}(u)=\sum_{i=1}^{k}{\omega_{h}(u_{i})}. $ Homogeneous distance may not be a metric, the following lemma gives a characterization of homogeneous distance to be a metric on $\mathbb{Z}_{n}$.
    \begin{Lemma}{\rm{(\cite{hc,ll})}}\label{Lem1}
		Homogeneous distance on $\mathbb{Z}_{n}$ is a metric if and only if $6 \nmid n$.
    \end{Lemma}
    According to Lemma \ref{Lem1}, we know that homogeneous distance is a metric on $ \mathbb{Z}_{2^{l}}. $

    Let $ r $ is a positive integer and $ \rho $ is a real number. We denote
    $$ B_{r}(\rho)=\{\, x \in \mathbb{Z}_{2^{l}}^{r} \, | \, \omega_{h}(x) \leqslant {r\rho} \, \}. $$
    The main theorem in \cite{l} used the entropy function, by Equation (5) in \cite{l}, this Theorem can be written as follows.
    \begin{Lemma}\label{lem2.3}{\rm (\cite{l})}
        Let $ \rho $ be a positive real number such that $ \omega_{min} < \rho \leqslant \bar{\omega}, $ where $ {\omega_{min}=\min_{a \in \mathbb{Z}_{2^{l}}}{\omega_{h}(a)}} $ and $ \bar{\omega}={|\mathbb{Z}_{2^{l}}|^{-1}\sum_{a \in \mathbb{Z}_{2^{l}}}{\omega_{h}{(a)}}} $ are the minimum weight and the average weight of elements of the residue ring $ \mathbb{Z}_{2^{l}}. $ Then $$ |B_{r}(\rho)| \leqslant {(e^{ \lambda\rho }\mathcal{L}(\lambda))^{r}} $$
        holds for all $ \lambda \in [0,\infty) $ and $ n \in \mathbb{N}, $ where $ \mathcal{L}(\lambda)=\sum_{a \in \mathbb{Z}_{2^{l}}}{e^{-\lambda\omega_{h}(a)}} $ and $ e $ is natural constant.
    \end{Lemma}

    \subsection{Function-correcting codes with homogeneous distance and irregular homogeneous distance codes}

     Let $ u \in \mathbb{Z}_{2^{l}}^{k} $ be a message and $ f:\mathbb{Z}_{2^{l}}^{k} \rightarrow {\rm Im}(f) \triangleq {\{ f(u):u \in \mathbb{Z}_{2^{l}}^{k} \}} $ be a function on $ \mathbb{Z}_{2^{l}}^{k}, $ where $ k $ is a positive integer. We define an encoding function $$ Enc:\mathbb{Z}_{2^{l}}^{k} \rightarrow \mathbb{Z}_{2^{l}}^{k+r},~Enc(u)=(u,p(u)), $$ where $ p(u) \in \mathbb{Z}_{2^{l}}^{r} $ is the redundancy vector of the codeword $ Enc(u) $ and $ r $ is the redundant length.

     Note that for each information $ u, $ the codeword $ Enc(u) $ is transmitted over an erroneous channel. In this paper, we assume that $ y \in \mathbb{Z}_{2^{l}}^{k+r} $ is a received word with $ d_{h}(Enc(u),y) \leqslant t . $ We have the following definition.

    \begin{Definition}\label{Def2}
       Suppose homogeneous distance between a sent codeword and a received codeword is at most $t$. An encoding function $Enc:\mathbb{Z}_{2^{l}}^{k} \rightarrow \mathbb{Z}_{2^{l}}^{k+r}$ with ${Enc(u)=(u,p(u))}, $ $ {u \in \mathbb{Z}_{2^{l}}^{k}} $ defines a function-correcting code with homogeneous distance on $ \mathbb{Z}_{2^{l}} $ for the function $ {f:\mathbb{Z}_{2^{l}}^{k} \rightarrow {\rm Im}(f)} $, denoted by {\rm FCCHD}, if for all $ u_{1},u_{2} \in \mathbb{Z}_{2^{l}}^{k}$ with $ f(u_{1}) \neq f(u_{2}) $ it holds that $$d_{h}(Enc(u_{1}),Enc(u_{2})) \geqslant 2t+1.$$
    \end{Definition}

    By this definition, the receiver can uniquely recover the correct $ f(u) , $ if the receiver knows the function $f(\cdot)$ and the encoding function $ Enc(\cdot) . $ The goal of constructing FCCHDs is to get correct $ f(u) $ for any message $ u \in \mathbb{Z}_{2^{l}}^{k} $ after transmitting.
    \begin{Remark}
        For a given function $ {f:\mathbb{Z}_{2^{l}}^{k} \rightarrow {\rm Im}(f)} $ and encoding function $ Enc $ of an {\rm{FCCHD}} $Enc:\mathbb{Z}_{2^{l}}^{k} \rightarrow \mathbb{Z}_{2^{l}}^{k+r}$, let $ {u \in \mathbb{Z}_{2^{l}}^{k}} $, suppose that $ y $ is received and $ {d_{h}(Enc(u),y) \leqslant t} $. Refer to highly similar decoding rules, using triangle inequality of homogeneous distance, the receiver can get correct function value $ f(u) $. Note that under these condition, the receiver may can not get the correct message $ u $.
    \end{Remark}

    We define the optimal redundancy of FCCHDs for a function $f$ as follows.
    \begin{Definition}\label{Def3}
       The optimal redundancy $r_{h}^{f}(k,t)$ is defined as the smallest integer $ r $ such that there exists an {\rm FCCHD} with encoding function $Enc:\mathbb{Z}_{2^{l}}^{k} \rightarrow \mathbb{Z}_{2^{l}}^{k+r}$ for the function $f$ on $ \mathbb{Z}_{2^{l}}^{k} $.
    \end{Definition}

    For convenience, we denote $[a]^{+} \triangleq \max\{a,0\}, $ for any integer $a.$ For any positive integer $M,$ let $[M] \triangleq \{1,2,\ldots,M\}$ be a set. For a matrix $ D , $ $[D]_{ij}$ is the $(i,j)$ entry of $D.$ For $ M>0 , $ $ \mathbb{N}_{0}^{M \times M} $ denotes the set of $ M \times M $ matrices with nonnegative integer entries.

    Let $ P=\{\,p_{1},p_{2},\ldots,p_{M}\,\} \subseteq \mathbb{Z}_{2^{l}}^{r} $ be a code of length $r$ and cardinality $ M . $ We define a class of codes with irregular homogeneous distance requirements as follows.
    \begin{Definition}\label{Def5}
       Let $ D \in \mathbb{N}_{0}^{M \times M} $ be a matrix. $ P=\{\,p_{1},p_{2},\ldots,p_{M}\,\} \subseteq \mathbb{Z}_{2^{l}}^{r}, $ if there exists an ordering of codewords of $P$ such that $d_{h}(p_{i},p_{j}) \geqslant [D]_{ij}$ for al $i,j \in [M], $ then we call $ P $ a $D$-homogeneous distance code.

       Furthermore, we define $N_{h}(D)$ to be the smallest integer $r$ such that there exists a $D$-homogeneous distance code of length $ r . $ In particular, if $[D]_{ij}=d,\,\forall{\,i \neq j},~i,j \in [M] , $ we write $N_{h}(D)$ as $N_{h}(M,d) . $
    \end{Definition}

    Next, we define homogeneous distance requirement matrices for a function $f$.
    \begin{Definition}\label{Def4}
       Let $ u_{1},u_{2},\ldots,u_{M} \in \mathbb{Z}_{2^{l}}^{k} $ be $ M $ vectors over $ \mathbb{Z}_{2^{l}}^{k} , $ where $ k,l,M $ are positive integers. We define homogeneous distance requirement matrix $ D_{h}^{f}(t,u_{1},u_{2},\ldots,u_{M}) $ for a function $f$ as the $ M \times M $ matrix with entries
       \begin{equation}\label{Formula 2}
           [D_{h}^{f}(t,u_{1},u_{2},\ldots,u_{M})]_{ij}=\left\{
  	         \begin{aligned}
  	          &[2t+1-d_{h}(u_{i},u_{j})]^{+},& \text{if } f(u_{i}) \neq f(u_{j}), \\
                &0,& \text{otherwise}.
               \end{aligned}
               \right.
        \end{equation}
    \end{Definition}

    \begin{Definition}\label{Def6}
      Let $ f:\mathbb{Z}_{2^{l}}^{k} \rightarrow {\rm Im}(f) $ be a function. For any $ z \in {\rm Im}(f), $ let $ f^{-1}(z) $ be the set of preimages of $ z, $ i.e. , $ f^{-1}(z)={ \{\, u \in \mathbb{Z}_{2^{l}}^{k} \, | \, f(u)=z \, \} }. $ For $ \forall z_{1}, \, z_{2} \in {\rm Im}(f), $ define function homogeneous distance between function values in $ {\rm Im}(f) $ as follows.
      $$ d_{h}^{f}(z_{1},z_{2})={ \min{ \{\, d_{h}(u,v) \, | \, u \in f^{-1}(z_{1}) , \, v \in f^{-1}(z_{2}) \, \} }}, $$
      where $ d_{h}^{f}(z_{1},z_{2})=0,\text{ if } z_{1}=z_{2} \in {\rm Im}(f). $
    \end{Definition}

    Using Definition \ref{Def6}, we give a definition of the function homogeneous distance matrix for a function $f$ as follows.
    \begin{Definition}\label{Def7}
      The function homogeneous distance matrix for a function $f$ is denoted by an $E \times E$ matrix $ D_{h}^{f}(t,f_{1},\ldots,f_{E}), $ whose entries are defined as follows.
       \begin{equation}\label{Formula 3}
           [D_{h}^{f}(t,f_{1},\ldots,f_{E})]_{ij}=\left\{
  	         \begin{aligned}
  	          &[2t+1-d_{h}^{f}(f_{i},f_{j})]^{+},& \text{if } i \neq j , \\
                &0,& \text{otherwise},
               \end{aligned}
               \right.
        \end{equation}
        where $E=|{\rm Im}(f)|,~{\rm Im}(f)=\{f_{1},f_{2},\ldots,f_{E}\}$.
    \end{Definition}
    %\begin{Remark}
     %   In this paper, we use some specific $D$-%homogeneous distance codes to get some bounds of %$r_{h}^{f}(k,t)$ for the function $f$, where matrices %$D$ mainly are the homogeneous distance requirement %matrix and the function homogeneous distance matrix %for the function\,$ f. $
    %\end{Remark}

\section{Bounds of FCCHDs for general functions}
	We find some connections between $r_{h}^{f}(k,t)$ of FCCHDs for a function $f$ and some specific $D$-homogeneous distance codes. In this section, we let $E=|{\rm Im}(f)|$ and we give some bounds of $r_{h}^{f}(k,t)$ using these connections.

    \subsection{Connections between FCCHDs and $ D $-homogeneous distance codes }
     In this subsection, we mainly use homogeneous distance requirement matrices and function homogeneous distance matrices $ D $ to construct $ D $-homogeneous distance codes. For a function $ f $, we associate redundancy of FCCHDs with $ D $-homogeneous distance codes. Using definitions of some special matrices, we find some connections between $ r_{h}^{f}(k,t) $ and $ N_{h}(D) $ which is the smallest length of $ D $-homogeneous distance codes.
	
	\begin{Theorem}\label{Th1}
		For any function $f:\mathbb{Z}_{2^{l}}^{k} \rightarrow {\rm Im}(f)$,
         \begin{equation}\label{Formula 5}
           r_{h}^{f}(k,t)=N_{h}(D_{h}^{f}(t,u_{1},u_{2},\ldots,u_{2^{lk}})),
        \end{equation}
        where $\{u_{1},u_{2},\ldots,u_{2^{lk}}\}=\mathbb {Z}_{2^{l}}^{k}.$
	\end{Theorem}

    \begin{proof}
       To prove $r_{h}^{f}(k,t) \geqslant N_{h}(D_{h}^{f}(t,u_{1},\ldots,u_{2^{lk}})), $ we assume the contrary that $$r_{h}^{f}(k,t)<N_{h}(D_{h}^{f}(t,u_{1},u_{2},\ldots,u_{2^{lk}})).$$
       Let the encoding function $Enc:\mathbb{Z}_{2^{l}}^{k} \rightarrow \mathbb{Z}_{2^{l}}^{k+r},u_{i} \mapsto (u_{i},p_{i})$ be an FCCHD with redundancy $ r=r_{h}^{f}(k,t). $ Since redundancies of this FCCHDs do not form a $ D_{h}^{f}(t,u_{1},u_{2},\ldots,u_{2^{lk}}) $-homogeneous distance code, there exists two redundancy vectors $p_{i},~p_{j},~i \neq j$ such that $$d_{h}(p_{i},p_{j})<2t+1-d_{h}(u_{i},u_{j}),$$ hence,
       $$ d_{h}(Enc(u_{i}),Enc(u_{j}))=d_{h}(u_{i},u_{j})+d_{h}(p_{i},p_{j})<2t+1, $$
       this is a contradiction to the definition of FCCHDs.

       On the other hand,\! suppose that $ {p_{1},p_{2},\ldots,p_{2^{lk}}} $\! form a\! $ {D_{h}^{f}(t,u_{1},u_{2},\ldots,u_{2^{lk}})} $-homogeneous distance code with length $ r'={N_{h}(D_{h}^{f}(t,u_{1},u_{2},\ldots,u_{2^{lk}}))}. $ Then the encoding function ${Enc:\mathbb{Z}_{2^{l}}^{k} \rightarrow \mathbb{Z}_{2^{l}}^{k+r'}},$ $u_{i} \mapsto (u_{i},p_{i})$ is an FCCHD by using Definition \ref{Def4}, thus $$r_{h}^{f}(k,t) \leqslant N_{h}(D_{h}^{f}(t,u_{1},u_{2},\ldots,u_{2^{lk}})).$$

       Therefore, $$ r_{h}^{f}(k,t)=N_{h}(D_{h}^{f}(t,u_{1},u_{2},\ldots,u_{2^{lk}})). $$
    \end{proof}

    Since\! ${N_{h}(D_{h}^{f}(t,u_{1},u_{2},\ldots,u_{2^{lk}}))}$\! is usually hard to compute,\! where\! ${\{u_{1},u_{2},\ldots,u_{2^{lk}}\}=\mathbb {Z}_{2^{l}}^{k}},$ we can choose less different vectors $ {u_{1},u_{2},\ldots,u_{M}} $ in $ \mathbb{Z}_{2^{l}}^{k} $ to obtain a simpler matrix $ {D_{h}^{f}(t,u_{1},u_{2},\ldots,u_{M})}, $ therein, $ M $ is a positive integer. This is easier for us to obtain a bound.
    \begin{Corollary}\label{Cor1}
       Let $ u_{1},u_{2},\ldots,u_{M} \in \mathbb{Z}_{2^{l}}^{k} $ be arbitrary different vectors. Then the optimal redundancy of {\rm FCCHDs} for any function $ f $ satisfies the following lower bound: $$ r_{h}^{f}(k,t) \geqslant N_{h}(D_{h}^{f}(t,u_{1},u_{2},\ldots,u_{M})). $$
       In particular, for any function $f$ with $|{\rm Im}(f)| \geqslant 2, $ then $ r_{h}^{f}(k,t) \geqslant t. $
    \end{Corollary}

    \begin{proof}
       Since $ \{u_{1},u_{2},\ldots,u_{M}\} \subseteq \mathbb{Z}_{2^{l}}^{k}=\{u_{1},u_{2},\ldots,u_{2^{lk}} \},$ by Theorem \ref{Th1}, we have $$ N_{h}(D_{h}^{f}(t,u_{1},u_{2},\ldots,u_{M})) \leqslant N_{h}(D_{h}^{f}(t,u_{1},u_{2},\ldots,u_{2^{lk}}))=r_{h}^{f}(k,t). $$

       If $|{\rm Im}(f)| \geqslant 2, $ we can claim that there exist $u, \, u' \in \mathbb{Z}_{2^{l}}^{k}$ with $ d_{h}(u,u')=1 $ satisfying $ f(u) \neq f(u'). $ In the following, we prove this claim. Suppose the contrary, for all $ u,\,u' \in \mathbb{Z}_{2^{l}}^{k} $ with $ d_{h}(u,u')=1, $ we have $ f(u)=f(u'). $ Firstly, we can fix a message $ u \in \mathbb{Z}_{2^{l}}^{k}, $ then function values of all messages in the ball $ B_{h}(u,1) $ are equal $ f(u), $ where $ {B_{h}(u,1)=\{\,v \in \mathbb{Z}_{2^{l}}^{k} \, | \, d_{h}(v,u) \leqslant 1 \,\}}. $ Next, we choose a different message $ u' $ in $ B_{h}(u,1), $ similarly, function values in the new ball $ B_{h}(u',1) $ are  equal $ f(u'). $ Continue this process, we obtain that $ |{\rm Im}(f)|=1, $ a contradiction. Then $$r_{h}^{f}(k,t) \geqslant N_{h}(D_{h}^{f}(t,u,u'))=N_{h}(2,2t).$$

       By Definition \ref{Def4}, $ P=\{(\overset{t}{\overbrace{0,0,\ldots,0}}),(\overset{t}{\overbrace{2^{l-1},2^{l-1},\ldots,2^{l-1}}})\} $ is a $D_{h}^{f}(t,u,u')$-homogeneous distance code, thus $N_{h}(2,2t)=t.$
    \end{proof}

    The purpose of an FCCHD for a function $ f $ is to get correct function values after transmitting, thus we need to study the connection between $r_{h}^{f}(k,t)$ and $N_{h}(D),$ where the matrix $D$ is the function homogeneous distance matrix for the function $f.$
    \begin{Theorem}\label{Th2}
       For any function $f:\mathbb{Z}_{2^{l}}^{k} \rightarrow {\rm Im}(f),$
         \begin{equation}\label{Formula 6}
           r_{h}^{f}(k,t) \leqslant N_{h}(D_{h}^{f}(t,f_{1},\ldots,f_{E})),
        \end{equation}
        where $E=|{\rm Im}(f)|$, $\,{\rm Im}(f)=\{\,f_{1},\,f_{2},\,\ldots,\,f_{E}\,\}$.
    \end{Theorem}

    \begin{proof}
       For convenience, denote $ r \triangleq N_{h}(D_{h}^{f}(t,f_{1},\ldots,f_{E})) . $ We choose $ P=\{p_{1},\ldots,p_{E}\}$ $ \subseteq \mathbb{Z}_{2^{l}}^{r} $ to be a $ D_{h}^{f}(t,f_{1},\ldots,f_{E}) $-homogeneous distance code such that $d_{h}(p_{i},p_{j}) \geqslant 2t+1-d_{h}^{f}(f_{i},f_{j}) , $ when $ i \neq j . $

       Define an encoding function $Enc:u \mapsto (u,p_{i}) , $ where $ f(u)=f_{i} . $ Thus, codewords of two information vectors with the same function value have the same redundancy vector. For any $ u_{i},u_{j} $ with $ f(u_{i})=f_{i},~f(u_{j})=f_{j}, $ if $ f_{i} \neq f_{j}, $ we have $$ d_{h}(Enc(u_{i}),Enc(u_{j}))=
               d_{h}(u_{i},u_{j})+d_{h}(p_{i},p_{j}). $$
        Refer to Definition \ref{Def5} and Definition \ref{Def6}, we obtain $$ d_{h}(Enc(u_{i}),Enc(u_{j}))\geqslant {d_{h}^{f}(f_{i},f_{j})+2t+1-d_{h}^{f}(f_{i},f_{j})}, $$ thus $ d_{h}(Enc(u_{i}),Enc(u_{j})) \geqslant {2t+1}. $
         By using the definition of FCCHDs, this encoding function $ Enc $ defines an FCCHD. Thus $$r_{h}^{f}(k,t) \leqslant N_{h}(D_{h}^{f}(t,f_{1},\ldots,f_{E})).$$
    \end{proof}

    There is an important special case such that the bound in Theorem \ref{Th2} to be tight as follows.
    \begin{Corollary}\label{Cor2}
        If there exists a set of representative information vectors $ u_{1},u_{2},\ldots,u_{E} $ with $ \{f(u_{1}),f(u_{2}),\ldots,f(u_{E})\}={\rm Im}(f) $ and $ D_{h}^{f}(t,f_{1},\ldots,f_{E})=D_{h}^{f}(t,u_{1},\ldots,u_{E}), $ then
        $$ r_{h}^{f}(k,t)=N_{h}(D_{h}^{f}(t,f_{1},\ldots,f_{E})).$$
    \end{Corollary}

    \begin{proof}
        By Corollary \ref{Cor1} and Theorem \ref{Th2}, we can easily obtain this result.
    \end{proof}

    \subsection{Bonds of $ N_{h}(D) $ for some matrices $ D $}
    Before this subsection, we have found some connections between the values of $ r_{h}^{f}(k,t) $ and $ N_{h}(D) $ for some special matrices $ D $. Thus if we know bounds of $ N_{h}(D) $, we will obtain some bounds of $ r_{h}^{f}(k,t) $. In this subsection, we mainly study bounds of $ N_{h}(D) $ for some matrices $ D $.

    \begin{Lemma}\label{Lem2}{\rm(\cite{go})}
       For any $(n,M,d)$ code $C$ with homogeneous distance, where $ d $ is the minimum homogeneous distance of code $ C, $ the following inequality holds
       $$ M(M-1)d \leqslant \sum_{x,y \in C}{\omega_{h}(x-y)} \leqslant nM^{2}. $$
    \end{Lemma}

    By using the above lemma and Definition \ref{Def5}, the following lemma is obtained.
    \begin{Lemma}\label{Lem3}
        For any matrix $D \in \mathbb{N}_{0}^{M \times M},$ $$ N_{h}(D) \geqslant \frac{1}{M^{2}}\sum_{i,j}{[D]_{ij}}. $$
    \end{Lemma}

    \begin{proof}
        Let $r=N_{h}(D)$, $P=\{p_{1},p_{2},\ldots,p_{M}\} \subseteq \mathbb{Z}_{2^{l}}^{r} $ be a $D$-homogeneous distance code of length $ r $. By Lemma \ref{Lem2}, then $$ \sum_{i,j}{\omega_{h}(p_{i}-p_{j})}=\sum_{i,j}{d_{h}(p_{i},p_{j})} \leqslant rM^{2}. $$ Since $$ \sum_{i,j}{[D]_{ij}} \leqslant \sum_{i,j}{d_{h}(p_{i},p_{j})} \leqslant rM^{2}, $$
        we have $$r=N_{h}(D) \geqslant \frac{1}{M^{2}}\sum_{i,j}{[D]_{ij}}.$$
    \end{proof}

    For any homogeneous distance requirement matrix $ D $, Lemma \ref{Lem3} offers a lower bound of $N_{h}(D)$ to be $ \frac{2}{M^{2}}\sum_{i,j:i<j}{[D]_{ij}} $ as $ D $ is a symmetric matrix. Next, we show an upper bound of the smallest length of $D$-homogeneous distance codes.

    Fix an element $ x $ in $ \mathbb{Z}_{2^{l}}^{r} , $ we denote
    $$ V_{h}(r,d) \triangleq |\{~c \in \mathbb{Z}_{2^{l}}^{r}~|~d_{h}(c,x) \leqslant d~\}|. $$ Similar to Gilbert-Varshamov bound, we obtain a lemma as follows.
    \begin{Lemma}\label{Lem4}
        For any homogeneous distance requirement matrix $D \in \mathbb{N}_{0}^{M \times M}$, and any permutation $\pi : [M] \rightarrow [M],$ $$ N_{h}(D) \leqslant \min_{r \in \mathbb{N}} { \left\{ ~ r:2^{lr} > \max_{j \in [M]}{\sum_{i=1}^{j-1}{V_{h}(r,{[D]_{\pi(i)\pi(j)}-1})}} ~ \right\} ~ }. $$
    \end{Lemma}

    \begin{proof}
        For simplicity, we assume $ \pi(i)=i,\,\forall \, i \in [M] . $
        Choosing an arbitrary codeword $ p_{1} \in \mathbb{Z}_{2^{l}}^{r} , $ if $ (2^{l})^{r}>V_{h}(r,{[D]_{12}-1}), $ then there exists $p_{2} \in \mathbb{Z}_{2^{l}}^{r}$ such that $ d_{h}(p_{1},p_{2}) \geqslant [D]_{12}. $ Similarly, if  $$(2^{l})^{r}>{V_{h}(r,{[D]_{13}-1})+V_{h}(r,{[D]_{23}-1})},$$ we can choose $ p_{3} \in \mathbb{Z}_{2^{l}}^{r} $ such that $d_{h}(p_{1},p_{3}) \geqslant [D]_{13}, ~ d_{h}(p_{2},p_{3}) \geqslant [D]_{23}.$ Continue this process, if $$(2^{l})^{r}>\sum_{i=1}^{M-1}{V_{h}(r,{[D]_{iM}-1})},$$ we can choose $p_{M} \in \mathbb{Z}_{2^{l}}^{r}$ such that $d_{h}(p_{i},p_{M}) \geqslant [D]_{iM}.$

        Therefore, if $ (2^{l})^{r}>\max_{j \in [M]}{\sum_{i=1}^{j-1}{V_{h}(r,{[D]_{ij}-1})}}, $ there exist $p_{1},~p_{2},~\ldots~,~p_{M}$ to be codewords of a $D$-homogeneous distance code. By Definition \ref{Def5}, codewords of $D$-homogeneous distance codes can be chosen in arbitrary order , then $$ N_{h}(D) \leqslant \min_{r \in \mathbb{N}}{\{r:2^{lr} > \max_{j \in [M]}{\sum_{i=1}^{j-1}{V_{h}(r,{[D]_{\pi(i)\pi(j)}-1})}}\}}. $$
    \end{proof}

    Since for any two vectors $x,~y \in \mathbb{Z}_{2^{l}}^{k},~d_{h}(x,y) \geqslant d(x,y),$ where $d(x,y)$ denotes Hamming distance between vectors $x,~y.$ In order to use this inequation to obtain a bound of $ N_{h}(D) $ for some special matrices $ D $, we need the following definition and lemmas.
    \begin{Definition}{\rm (\cite{lbwy})}\label{Def6.1}
        Let $ D \in \mathbb{N}_{0}^{M \times M} . $ Then $ P=\{p_{1},p_{2},\ldots,p_{M}\} $ is a $D$-code, if there exists an ordering of codewords of $P$ such that $d(p_{i},p_{j}) \geqslant [D]_{ij}$ for all $i,j \in [M], $ where $ d(p_{i},p_{j}) $ denotes the Hamming distance between $ p_{i} $ and $ p_{j}. $

       Furthermore, we define $N(D)$ to be the smallest integer $r$ such that there exists a $D$-code of length $ r . $ If $[D]_{ij}=d$ for all $i \neq j,~i,j \in [M] , $ we write $N(D)$ to be $N(M,d) . $
    \end{Definition}

    In {\cite{lbwy}}, it seems that there is a small error in Lemma 4 due to a computation error in the proof. The correct one is as follows.
    \begin{Lemma}{\rm (\cite{lbwy})}\label{Lem6.1}
        For any $ M,d \in \mathbb{N} $ with $ d \geqslant 10 $ and $ M \leqslant d^{2} , $
        \begin{equation}\label{2.3}
            N(M,d) \leqslant \frac{2d-2}{1-2\sqrt{{\ln{(d)}}/{d}}}.
        \end{equation}
    \end{Lemma}
    \begin{Lemma}\label{Lem5}
        For any $d,M \in \mathbb{N},~N_{h}(M,d) \leqslant N(M,d).$
    \end{Lemma}

    \begin{proof}
        By Definition \ref{Def6.1}, let $r=N(M,d),$ then there exists a code $C=\{p_{1},p_{2},\ldots,p_{M}\}$ of length $r$ to be a $D$-code such that $d(p_{i},p_{j}) \geqslant d$ for all $i \neq j,~i,j \in [M].$ Since $$d_{h}(p_{i},p_{j}) \geqslant d(p_{i},p_{j}) \geqslant d , $$ $C$ is a $D$-homogeneous distance code, hence $$N_{h}(M,d) \leqslant N(M,d).$$
    \end{proof}

    Following Lemma \ref{Lem6.1} and Lemma \ref{Lem5}, we obtain another upper bound of $N_{h}(M,d)$ as follows.
    \begin{Proposition}\label{Lem6}
        For any $M,d \in \mathbb{N}$ with $d \geqslant 10$ and $ M \leqslant d^{2} , $ then $$ N_{h}(M,d) \leqslant \frac{2d-2}{1-2\sqrt{{\ln{(d)}}/{d}}} . $$
    \end{Proposition}

    \begin{proof}
        By using Inequation (\ref{2.3}), it follows from Lemma \ref{Lem5} that $$N_{h}(M,d) \leqslant N(M,d) \leqslant \frac{2d-2}{1-2\sqrt{{\ln{(d)}}/{d}}} . $$
    \end{proof}

    By using Lemma \ref{lem2.3} and Lemma \ref{Lem4}, we obtain another bound of $ N_{h}(M,d) $ without restrictions on the parameters.
    \begin{Theorem}\label{Th3}
        For any $ M,d \in \mathbb{N}, $ it holds that
        $$ N_{h}(M,d) \leqslant \left\lceil {\frac{\ln{M}+d-1}{1-\ln{2}}} \right\rceil. $$
    \end{Theorem}

    \begin{proof}
        By Lemma \ref{Lem4}, we have
        \begin{equation*}
            \begin{aligned}
                N_{h}(M,d) &\leqslant {\min_{r \in \mathbb{N}} { \left\{ ~ r:2^{lr} > \max_{j \in [M]}{\sum_{i=1}^{j-1}{V_{h}(r,{d-1})}} ~ \right\} ~ }},\\
                &\leqslant {\min_{r \in \mathbb{N}}{\{~r:2^{lr}>MV_{h}(r,{d-1})\}}}.
            \end{aligned}
        \end{equation*}
        Using Lemma \ref{lem2.3},
        $$ V_{h}(r,d-1)=\left |B_{r}(\frac{d-1}{r}) \right | \leqslant {\left (e^{\lambda\frac{d-1}{r}}\sum_{a \in \mathbb{Z}_{2^{l}}}{e^{-\lambda\omega_{h}(a)}}\right )^{r}}. $$
        Let $ \lambda=1, $ then
        $$ V_{h}(r,d-1) \leqslant {e^{d-1}\left (\sum_{a \in \mathbb{Z}_{2^{l}}}{e^{-\omega_{h}(a)}}\right )^{r}} \leqslant {e^{d-1}[(2^{l}-1)e^{-1}+1]^{r}} \leqslant {e^{d-r-1}2^{(l+1)r}}. $$

        If $ {2^{lr} > {e^{d-r-1}2^{(l+1)r}M}}, $ then $ {2^{lr} > {MV_{h}(r,{d-1})}}. $ We only need to compute the inequation $ {{e^{d-r-1}2^{(l+1)r}M} < 2^{lr}} $. If $ {{e^{d-r-1}2^{(l+1)r}M} < 2^{lr}} $, we have $$ {{2^{r}Me^{d-r-1}} < 1}. $$ By logarithmic operation, we have $$ {{\ln{(2^{r}e^{-r})}} < \ln{(Me^{d-1})^{-1}}}, $$ hence we get $$ {r > {\frac{\ln{M}+(d-1)}{1-\ln2}}}. $$ That means, if $ r=\left\lceil {\frac{\ln{M}+d-1}{1-\ln{2}}} \right\rceil, $ we have $ r \geqslant {d-1} $ and $ 2^{lr} \geqslant {MV_{h}(r,{d-1})}. $

        Therefore, $ N_{h}(M,d) \leqslant \left\lceil {\frac{\ln{M}+d-1}{1-\ln{2}}} \right\rceil. $

    \end{proof}

    \begin{Remark}
        Proposition \ref{Lem6} and Theorem \ref{Th3} give two upper bounds of $ N_{h}(M,d), $ while Proposition \ref{Lem6} gives a bound when $ d \geqslant 10 $ and $ M \leqslant d^{2}. $ Thus we can not compare these two bounds directly. If we restrict $ d $ and $ M $ in Theorem \ref{Th3}, we can compare these two bounds. If $ {d \geqslant 10},\, { M\leqslant d^{2} }, $ then $ {\left\lceil {\frac{\ln{M}+d-1}{1-\ln{2}}} \right\rceil} \leqslant {\left\lceil {\frac{2\ln{d}+d-1}{1-\ln{2}}} \right\rceil}. $ We now compare $ f(d)={\left\lceil {\frac{2\ln{d}+d-1}{1-\ln{2}}} \right\rceil} $ and $ g(d)=\frac{2d-2}{1-2\sqrt{{\ln{(d)}}/{d}}}. $ Note that under the condition of $ { M\leqslant d^{2} }, $ if $ {f(d) \leqslant g(d)}, $ then $ {\left\lceil {\frac{\ln{M}+d-1}{1-\ln{2}}} \right\rceil} \leqslant \frac{2d-2}{1-2\sqrt{{\ln{(d)}}/{d}}} $ holds. By observing Figure 1, we conclude that when $ {10 \leqslant d \leqslant 93},~M \leqslant d^{2}, $ the bound in Theorem \ref{Th3} is better than the bound in Proposition \ref{Lem6}.
    \end{Remark}

    \begin{figure}[htb]
      \centering
      \includegraphics[width=0.5\textwidth]{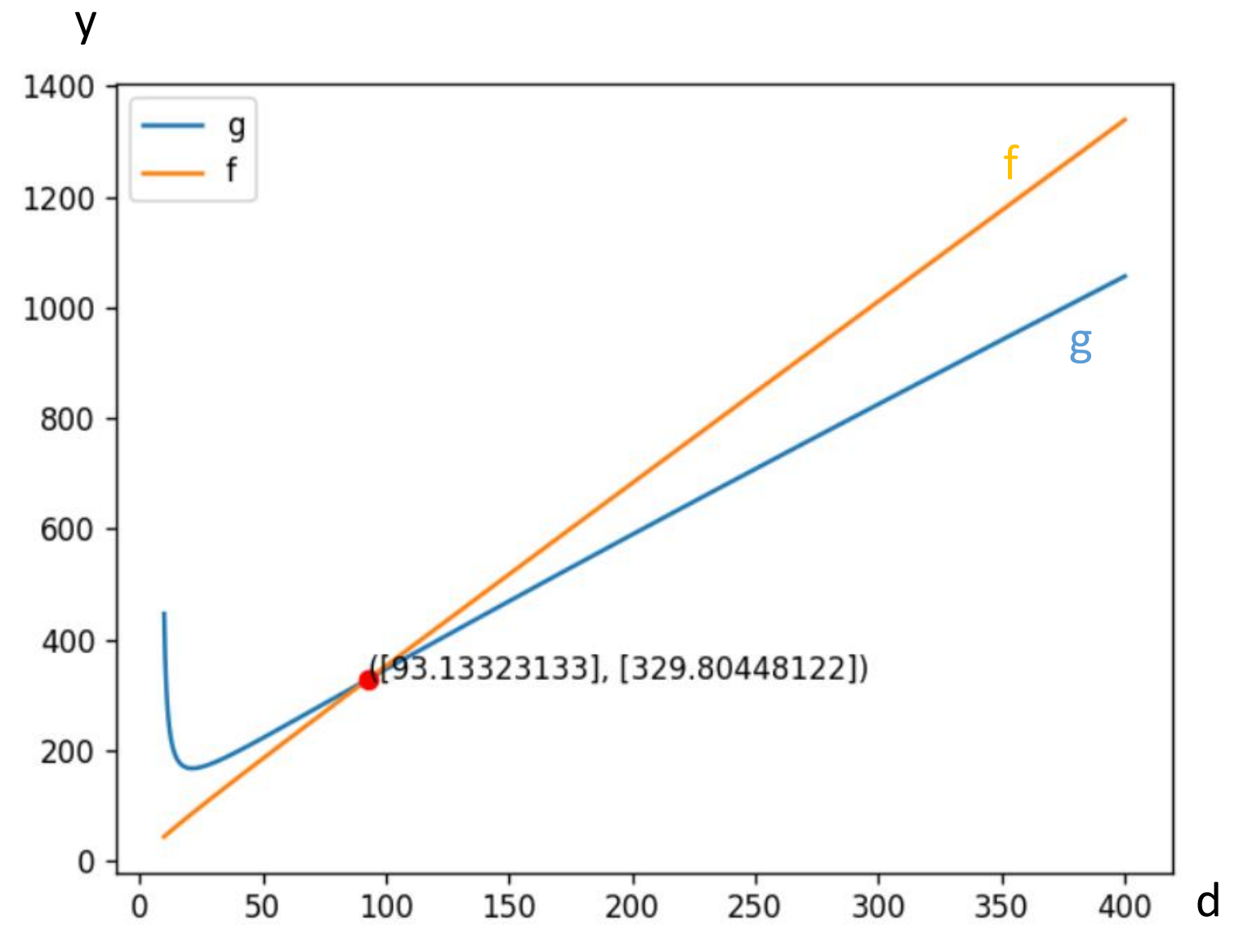}
      \caption[]{Comparision of functions $ f(d)={\left\lceil {\frac{2\ln{d}+d-1}{1-\ln{2}}} \right\rceil} $ and $ g(d)=\frac{2d-2}{1-2\sqrt{{\ln{(d)}}/{d}}}. $}
      \label{fig:pic1}
    \end{figure}

    We can get some values of $ {\left\lceil {\frac{\ln{M}+d-1}{1-\ln{2}}} \right\rceil} $ and $ \frac{2d-2}{1-2\sqrt{{\ln{(d)}}/{d}}} $ for some positive integers $ M,d $ as follows.
    \begin{table}[htb]
    \caption{Some values of upper bounds in Lemma \ref{Lem6} and Theorem \ref{Th3}}
    \centering
    \scriptsize
    \begin{tabular}{p{1cm}p{1cm}p{5.5cm}p{5.5cm}}
        \hline
        $d$ & $M$ &        $\frac{2d-2}{1-2\sqrt{{\ln{(d)}}/{d}}}$   & $\left\lceil {\frac{\ln{M}+d-1}{1-\ln{2}}} \right\rceil$ \\
        \hline
        10   & 90       & 446.7075690846029  & 44\\
        50   & 200      & 222.43907485184036  & 177\\
        70   & 200      & 272.0381083046102      & 243 \\
        80  & 300      & 297.0386071011767     & 277 \\
        90  & 400      & 321.9995158887893    & 310  \\
        \hline
    \end{tabular}
     \label{table-2}
    \end{table}

    \section{Bounds of FCCHDs for functions based on weights}
    In this section, we study FCCHDs for homogeneous weight functions, homogeneous weight distribution functions and Rosenbloom-Tsfasnman weight functions. We also offer two encoding functions which are proved to be FCCHDs for homogeneous weight functions and homogeneous weight distribution functions, respectively. What's more, under some conditions, redundancies of some of these FCCHDs reach the optimal redundancy bounds for these functions, respectively.

    \subsection{FCCHDs for Homogeneous weight functions}

    Let $ f(u)=\omega_{h}(u) ,$ where $ {u \in \mathbb{Z}_{2^{l}}^{k}} $ and $ \omega_{h} $ is homogeneous weight function on $ \mathbb{Z}_{2^{l}}^{k}$. Denote $ {E'\triangleq {|{\rm Im}(\omega_{h})|-1}} ,$ by Equation (\ref{Formula}), we have $ {{\rm Im}(\omega_{h})=\{0,1,\ldots,2k\}}, $ $ {|{\rm Im}(\omega_{h})|={2k+1}} $ and $ {E'=2k}. $ Suppose that $ {{\rm Im}(\omega_{h})=\{f_{0},\ldots,f_{E'}\}}. $ In this part, for convenience, we write $ D_{h}^{\omega_{h}}(t) \triangleq D_{h}^{\omega_{h}}(t,f_{0},\ldots,f_{E'}). $

    Note that the bound in Theorem \ref{Th2} is tight for the  homogeneous weight function.

    \begin{Theorem}\label{Lem8}
    	Let $f(u)=\omega_{h}(u).$ Consider the $ (2k+1) \times (2k+1)$ matrix $D_{h}^{\omega_{h}}(t)$ with entries ${[D_{h}^{\omega_{h}}(t)]_{(i+1)(i+1)}=0} $ and $ {[D_{h}^{\omega_{h}}(t)]_{(i+1)(j+1)}=[2t+1-|i-j|]^{+} }$ for $ i \neq j,$ ${i,j \in \{\,0,1,\ldots,E'\,\}}. $ Then $$  r_{h}^{\omega_{h}}(k,t)=N_{h}(D_{h}^{\omega_{h}}(t)). $$
    \end{Theorem}

    \begin{proof}
        Let $ {{\rm Im}(\omega_{h})=\{0,1,2,\ldots,2k\}=\{f_{0},\ldots,f_{E'}\}}, $ where $ {E'=|{\rm Im}(\omega_{h})|-1}. $ We assume that ${f_{i}=i,\, {i \in \{\, 0,1,\ldots,E'\,\}}} . $ Thus $ {[D_{h}^{\omega_{h}}(t)]_{(i+1)(j+1)}=[2t+1-d_{h}^{f}(f_{i},f_{j})]^{+}},\,{i \neq j}. $ By Definition \ref{Def6}, $ {[D_{h}^{\omega_{h}}(t)]_{(i+1)(j+1)}=[2t+1-|i-j|]^{+}} ,\,{i \neq j},\,{i,j \in \{\,0,1,\ldots,E'\,\}}. $ By Theorem \ref{Th2}, we obtain that $ {r_{h}^{\omega_{h}}(k,t) \leqslant N_{h}(D_{h}^{\omega_{h}}(t))}. $

        On the other hand, let $ u_{i}=(\,\overset{m_{i}}{\overbrace{1,\,\cdots,\,1}}\, ~\, \overset{s_{i}}{\overbrace{2^{l-1},\,\cdots,\,2^{l-1}}}\, ~\, \overset{k-m_{i}-s_{i}}{\overbrace{0,\,\cdots,\,0}} \,),\,i=0,1,\ldots,E' , $ where $$ m_{i}=\max{\{\,m\,|\,m+2s=i,~m,s \in \mathbb{N}\,\}} , $$ and $$ s_{i}=\min{\{\,s\,|\,m+2s=i,~m,s \in \mathbb{N}\,\}} . $$ If $ s_{i}>0 $ and $ k-m_{i}-s_{i}>0 , $ we can suppose that $ m_{0}=m_{i}+2,~s_{0}=s_{i}-1 , $ then $ m_{0}+2s_{0}=i $ with $ m_{0}>m_{i} . $ This is a contradiction, thus, if $ s_{i}>0, $ $ k-m_{i}-s_{i}=0. $

        Without loss of generality, assume that $ m_{i} \geqslant m_{j} . $ We have
        $$ d_{h}(u_{i},u_{j})={(m_{i}-m_{j})+2\,|\,s_{i}-[s_{j}-(m_{i}-m_{j})]^{+}\,|}. $$ If $ s_{j}=0 , $ then $$ d_{h}(u_{i},u_{j})={(m_{i}-m_{j})+2s_{i}}={i-j};$$ If $ s_{j} \neq 0 , $ then $ s_{j}+m_{j}=k , $ and so
         \begin{equation*}
            \begin{aligned}
             d_{h}(u_{i},u_{j})=&
               {(m_{i}-m_{j})+2\,|\,s_{i}-[s_{j}-m_{i}+m_{j}]^{+}\,|\,} \\
               =&{(m_{i}-m_{j})+2\,|\,s_{i}-k+m_{i}\,|\,}\\
               =&{m_{i}-m_{j}+2k-2m_{i}-2s_{i}}\\
               =&{2k-m_{j}-(2s_{i}+m_{i})}\\
               =&{j-i} .
            \end{aligned}
         \end{equation*}
         Thus $ d_{h}(u_{i},u_{j})=|i-j|,$ where $ u_{i},~u_{j} $ satisfy $\omega_{h}(u_{i})=i,~\omega_{h}(u_{j})=j . $ Since $$ [D_{h}^{\omega_{h}}(t,u_{0},u_{1},\ldots,u_{E'})]_{(i+1)(j+1)}=[2t+1-|i-j|]^{+},\,{i \neq j},\,{i,j \in \{\, 0,1,\ldots,E'\,\}}, $$ we have $$ D_{h}^{\omega_{h}}(t,u_{0},u_{1},\ldots,u_{E'})=D_{h}^{\omega_{h}}(t) . $$
         By Corollary \ref{Cor2}, $$ r_{h}^{\omega_{h}}(k,t)=N_{h}( D_{h}^{\omega_{h}}(t,u_{0},u_{1},\ldots,u_{E'}))=N_{h}(D_{h}^{\omega_{h}}(t)) . $$
    \end{proof}

    \begin{Remark}
        Since homogeneous weight is a generalization of Hamming weight, Theorem \ref{Lem8} can be easily proved correct for the Hamming weight function.
    \end{Remark}

    Now, we give some examples.
    \begin{Example}
        For $ k=3,~t=2 , $ the function homogeneous distance matrix $ D_{h}^{\omega_{h}}(2) $ is given by the symmetric $ 7 \times 7 $ matrix, where
        $$ D_{h}^{\omega_{h}}(2)=\begin{pmatrix}
            0 & 4 & 3 & 2 & 1 & 0 & 0\\
            4 & 0 & 4 & 3 & 2 & 1 & 0\\
            3 & 4 & 0 & 4 & 3 & 2 & 1\\
            2 & 3 & 4 & 0 & 4 & 3 & 2\\
            1 & 2 & 3 & 4 & 0 & 4 & 3\\
            0 & 1 & 2 & 3 & 4 & 0 & 4\\
            0 & 0 & 1 & 2 & 3 & 4 & 0
        \end{pmatrix} . $$

        For $t=2,$ the function homogeneous distance matrix on $\mathbb{Z}_{2^{3}}^{4}$ of the homogeneous wight function is $$ D_{h}^{\omega_{h}}(2)=\begin{pmatrix}
            0 & 4 & 3 & 2 & 1 & 0 & 0 & 0 & 0\\
            4 & 0 & 4 & 3 & 2 & 1 & 0 & 0 & 0\\
            3 & 4 & 0 & 4 & 3 & 2 & 1 & 0 & 0\\
            2 & 3 & 4 & 0 & 4 & 3 & 2 & 1 & 0\\
            1 & 2 & 3 & 4 & 0 & 4 & 3 & 2 & 1\\
            0 & 1 & 2 & 3 & 4 & 0 & 4 & 3 & 2\\
            0 & 0 & 1 & 2 & 3 & 4 & 0 & 4 & 3\\
            0 & 0 & 0 & 1 & 2 & 3 & 4 & 0 & 4\\
            0 & 0 & 0 & 0 & 1 & 2 & 3 & 4 & 0
        \end{pmatrix} . $$
    \end{Example}

    Based on Theorem \ref{Lem8}, we get a lower bound of $ r_{h}^{\omega_{h}}(k,t) $ by using Lemma \ref{Lem3}, and the value of this bound is only depend on $t.$
    \begin{Corollary}\label{Cor3}
        For any $ k>{\lceil{\frac{t+1}{2}}\rceil}, $ $$ r_{h}^{\omega_{h}}(k,t) \geqslant \left\lceil {\frac{{5t^{3}+15t^{2}+10t}}{3(t+2)^{2}}}\right\rceil. $$
    \end{Corollary}

    \begin{proof}
        Let $ \{p_{1},p_{2},\ldots,p_{2k+1}\} \subseteq \mathbb{Z}_{2^{l}}^{r} $ be a $ D_{h}^{\omega_{h}}(t) $-homogeneous distance code. Consider the first $ t+2 $ codewords $ p_{1},~p_{2},\ldots,~p_{t+2} . $ By Lemma \ref{Lem3} and Theorem \ref{Lem8}, $$ r_{h}^{\omega_{h}}(k,t)=N_{h}(D_{h}^{\omega_{h}}(t)) \geqslant \frac{2}{(t+2)^{2}}\sum_{i,j:{0<i<j \leqslant {t+2}}}{[D_{h}^{\omega_{h}}(t)]_{ij}}, $$
        Following the proof of Theorem \ref{Lem8}, we have $$ r_{h}^{\omega_{h}}(k,t) \geqslant {\frac{2}{(t+2)^{2}}\sum_{m=0}^{t}{(2t-m)(t+1-m)}}. $$ Thus it holds that $ r_{h}^{\omega_{h}}(k,t) \geqslant \frac{{5t^{3}+15t^{2}+10t}}{3(t+2)^{2}}. $

        Since $ r_{h}^{\omega_{h}}(k,t) $ is an integer, we have $r_{h}^{\omega_{h}}(k,t) \geqslant \left\lceil {\frac{{5t^{3}+15t^{2}+10t}}{3(t+2)^{2}}}\right\rceil.$
    \end{proof}

    In the following, we provide an FCCHD for the homogeneous weight function on $ \mathbb{Z}_{2^{l}}^{k}. $ The shifted modulo operator of integer $a,b$ \rm{(\cite{lbwy})} is defined as
        \begin{equation}
            a \text{ smod } b \triangleq {((a-1)\mod{b})+1} \in \{1,2,\ldots,b\}.
        \end{equation}
    For examples, $ 2 \text{ smod } 5=2;~5 \text{ somd } 5=5. $

    \begin{Theorem}\label{Con1}
        Let $ \omega_{h} $ be the homogeneous weight function on $ \mathbb{Z}_{2^{l}}^{k} $. Then the encoding function $ Enc_{\omega_{h}}:\mathbb{Z}_{2^{l}}^{k} \rightarrow \mathbb{Z}_{2^{l}}^{k+r} $ is an {\rm{FCCHD}},
        $$ Enc_{\omega_{h}}(u)=(u,p_{\omega_{h}(u)+1}), $$
        where $ p_{i}, ~ i=1,2,\ldots,{2k+1}, $ are defined as follows.

        For $ t=1 , $ set $ p_{1}=(\,0 \ 0\,),~p_{2}=(\,1 \ 1\,) $ and $ p_{3}=(\,2^{l-1} \ 0\,). $ Then set $ p_{i}=p_{i \text{ smod }3} $ for $ i \geqslant 4.$

        For $t=2 , $ set $\, {p_{1}=(\,0 \ 0 \ 0\,)},~{p_{2}=(\,2^{l-1} \ 2^{l-1} \ 0\,)},\,{p_{3}=(\,0 \ 2^{l-1} \ 2^{l-1})},~{p_{4}=(2^{l-1} \ 0 \ 2^{l-1})} $. When $\, {5 \leqslant i \leqslant 8} $, let $\, p_{i}=p_{i-4}+(0 \ 1 \ 0) $, i.e., $ p_{5}=(\,0 \ 1 \ 0\,), $ $ p_{6}=(\,2^{l-1} \ {2^{l-1}+1} \ 0\,),$ $ p_{7}=(\,0 \ {2^{l-1}+1} \ 2^{l-1}\,),~p_{8}=(\,2^{l-1} \ 1 \ 2^{l-1}\,)$. When $\, i \geqslant 9 $, set $\, p_{i}=p_{i \text{ smod }8} $.

        For $ t=3, $ we set $~ p_{1}=(\ 0 \,\ 0 \,\ 0 \,\ 0 \,\ 0\ )\,,~p_{2}=(\ 1 \,\ 1 \,\ 1 \,\ 1 \,\ 2^{l-1}\ )\,,~p_{3}=(\ 2^{l-1} \,\ 2^{l-1} \,\ 2^{l-1} \,\ 0 \,\ 0\ )\,,\\p_{4}=(\,0 \ 0 \ 2^{l-1} \ 2^{l-1} \ 0\,),~p_{5}=(\,1 \ 1 \ {2^{l-1}+1} \ {2^{l-1}+1} \ 2^{l-1}\,),~p_{6}=(\,2^{l-1} \ 2^{l-1} \ 0 \ 2^{l-1} \ 0\,)$ and $p_{7}=(\,{2^{l-1}+1} \ {2^{l-1}+1} \ 2^{l-1} \ {2^{l-1}+1} \ 2^{l-1}\,). $ Then set $ p_{i}=p_{i \text{ smod }7} $ for $ i \geqslant 8. $

        For $ t > 3 , $ let $ p_{1},p_{2},\ldots,p_{2t+1} $ be a code with the minimum homogeneous distance $ 2t , $ i.e. $ d_{h}(p_{i},p_{j}) \geqslant 2t $ for all $ i,j \leqslant 2t+1,~i \neq j $ and set $ p_{i}=p_{i \text{ smod }(2t+1)} $ for $ i \geqslant {2t+2}. $
    \end{Theorem}

    \begin{proof}
        {\bf Case 1.}
        For $ t=1 , $ $ \forall{u_{1},u_{2}} \in \mathbb{Z}_{2^{l}}^{k} $ with $ \omega_{h}(u_{1}) \neq \omega_{h}(u_{2}) . $ Without loss of generality, we suppose $ {\omega_{h}(u_{1})=3m_{1}+\omega_{1}} $ and $\omega_{h}(u_{2})=3m_{2}+\omega_{2},$ where $\omega_{1},\omega_{2} \in \{0,1,2\}. $ Thus $ p_{\omega_{h}(u_{1})+1}=p_{\omega_{1}+1},~p_{\omega_{h}(u_{2})+1}=p_{\omega_{2}+1}. $

        If $ |{\omega_{h}(u_{1})-\omega_{h}(u_{2})}| \geqslant 3 , $ then $$ d_{h}(Enc_{\omega_{h}}(u_{1}),Enc_{\omega_{h}}(u_{2})) \geqslant d_{h}(u_{1},u_{2})=\omega_{h}(u_{1}-u_{2}) \geqslant |{\omega_{h}(u_{1})-\omega_{h}(u_{2})}| \geqslant 3. $$ If $ |{\omega_{h}(u_{1})-\omega_{h}(u_{2})}| < 3 , $ then $ p_{\omega_{1}+1} \neq p_{\omega_{2}+1} . $ Since $$ d_{h}(p_{1},p_{2})=d_{h}(p_{1},p_{3})=d_{h}(p_{2},p_{3})=2 , $$ we have $ d_{h}(p_{\omega_{1}+1},p_{\omega_{2}+1})=2 . $ Thus,
        $$ d_{h}(Enc_{\omega_{h}}(u_{1}),Enc_{\omega_{h}}(u_{2}))={d_{h}(u_{1},u_{2})+d_{h}(p_{\omega_{1}+1},p_{\omega_{2}+1})} \geqslant 3. $$
        Therefore, $ Enc_{\omega_{h}}:u \mapsto (u,p_{\omega_{h}(u)+1}) $ is an FCCHD, when $ t=1. $

        {\bf Case 2.}
        For $ t=2 , $ due to the definition of homogeneous distance, compute homogeneous distance between these $ p_{i}. $ For $ \forall{u_{1},u_{2}} \in \mathbb{Z}_{2^{l}}^{k} $ with $ \omega_{h}(u_{1}) \neq \omega_{h}(u_{2}), $ suppose that $ \omega_{h}(u_{1})={8m_{1}+\omega_{1}},$ $ \omega_{h}(u_{2})={8m_{2}+\omega_{2}}, $ then $ p_{\omega_{h}(u_{1})+1}=p_{\omega_{1}+1},~p_{\omega_{h}(u_{2})+1}=p_{\omega_{2}+1}. $ Since homogenous weight is a metric on $ \mathbb{Z}_{2^{l}}, $ we conclude that if $ |{\omega_{h}(u_{1})-\omega_{h}(u_{2})}|=j,\, 1 \leqslant j \leqslant 4, $ it holds that there exists $ i \in \{\, 1,2,\ldots,{8-j} \,\} $ such that $$ d_{h}(p_{\omega_{1}+1},p_{\omega_{2}+1})=d_{h}(p_{i},p_{i+j}) \geqslant {5-j}, $$ or $$ d_{h}(p_{\omega_{1}+1},p_{\omega_{2}+1})=d_{h}(p_{m+j-8},p_{m}) \geqslant {5-j}, $$ for some $ m \in \{\, {9-j}, \ldots , 8 \,\}. $ Thus
        \begin{equation*}
          \begin{aligned}
              d_{h}(Enc_{\omega_{h}}(u_{1}),Enc_{\omega_{h}}(u_{2}))&=d_{h}(u_{1},u_{2})+d_{h}(p_{\omega_{1}+1},p_{\omega_{2}+1})\\
              &\geqslant {|\omega_{h}(u_{1})-\omega_{h}(u_{2})|+d_{h}(p_{\omega_{1}+1},p_{\omega_{2}+1})}\\
              &\geqslant 5.
          \end{aligned}
        \end{equation*}
        If $ |{\omega_{h}(u_{1})-\omega_{h}(u_{2})}| \geqslant 5, $ $ d_{h}(Enc_{\omega_{h}}(u_{1}),Enc_{\omega_{h}}(u_{2})) \geqslant d_{h}(u_{1},u_{2}) \geqslant 5. $ Therefore, this code is an FCCHD.

        {\bf Case 3.}
        For $ t=3. $ We prove the above encoding function is an FCCHD when $ t=3, $ by Definition \ref{Def2}, we need to prove that if $ \omega_{h}(u_{1}) \neq \omega_{h}(u_{2}), $ for $ u_{1},u_{2} \in \mathbb{Z}_{2^{l}}^{k}, $ it holds that $$ d_{h}(Enc_{\omega_{h}}(u_{1}),Enc_{\omega_{h}}(u_{2}))={d_{h}(u_{1},u_{2})+d_{h}(p_{\omega_{h}(u_{1})+1},p_{\omega_{h}(u_{2})+1})} \geqslant 7. $$
        Similarly, if $ |{\omega_{h}(u_{1})-\omega_{h}(u_{2})}| \geqslant 7, $ $ d_{h}(u_{1},u_{2}) \geqslant {|{\omega_{h}(u_{1})-\omega_{h}(u_{2})}|}, $ obviously.

        If $ |{\omega_{h}(u_{1})-\omega_{h}(u_{2})}| < 7, $ assume that $ \omega_{h}(u_{1})=7m_{31}+\omega_{31},~\omega_{h}(u_{2})=7m_{32}+\omega_{32}, $ we have $ p_{\omega_{h}(u_{1})+1}=p_{\omega_{31}+1} $ and $ p_{\omega_{h}(u_{2})+1}=p_{\omega_{32}+1},\, {0 \leqslant \omega_{31},\omega_{32} \leqslant 6} . $ If $ {|{\omega_{h}(u_{1})-\omega_{h}(u_{2})}|=j}, $ $ 1 \leqslant j \leqslant 6, $ it is easy to prove that $ {{-1} \leqslant {m_{31}-m_{32}} \leqslant {1}},$ $ {{-7} < {\omega_{31}-\omega_{32} < 7}}. $ Due to the symmetry of the homogeneous distance, by Equation (\ref{Formula}), $$ d_{h}(p_{\omega_{31}+1},p_{\omega_{32}+1})=d_{h}(p_{i},p_{i+j}) \geqslant {7-j}, $$ for some $ i \in \{1,2,\ldots,{7-j}\}, $ or $$ d_{h}(p_{\omega_{31}+1},p_{\omega_{32}+1})=d_{h}(p_{m},p_{m+j-7}) \geqslant {7-j}, $$ for some $ m \in \{{8-j},\ldots,7\}. $ Thus $ d_{h}(Enc_{\omega_{h}}(u_{1}),Enc_{\omega_{h}}(u_{2})) \geqslant 7, $ for $ \forall u_{1},u_{2} \in \mathbb{Z}_{2^{l}}^{k} $ with $ \omega_{h}(u_{1}) \neq \omega_{h}(u_{2}). $ Therefore, this encoding function is an FCCHD when $ t=3. $

        {\bf Case 4.}
        For $ t > 3. $ $ \forall{u_{1},~u_{2}} \in \mathbb{Z}_{2^{l}}^{k} $ with $ \omega_{h}(u_{1}) \neq \omega_{h}(u_{2}). $ Without loss of generality, suppose that $ \omega_{h}(u_{1})={(2t+1)m_{1}'+\omega_{1}'}, $ $ \omega_{h}(u_{2})={(2t+1)m_{2}'+\omega_{2}'},~\omega_{1}',\omega_{2}' \in \{0,1,\ldots,2t\} . $ Thus we have $ {p_{\omega_{h}(u_{1})+1}=p_{\omega_{1}'+1},~ p_{\omega_{h}(u_{2})+1}=p_{\omega_{2}'+1}}. $

        If $ |{\omega_{h}(u_{1})-\omega_{h}(u_{2})}| \geqslant {2t+1} , $ then $$ d_{h}(Enc_{\omega_{h}}(u_{1}),Enc_{\omega_{h}}(u_{2})) \geqslant d_{h}(u_{1},u_{2})=\omega_{h}(u_{1}-u_{2}) \geqslant |{\omega_{h}(u_{1})-\omega_{h}(u_{2})}| \geqslant {2t+1}. $$ If $ |{\omega_{h}(u_{1})-\omega_{h}(u_{2})}| < {2t+1} , $ then $ p_{\omega_{1}'+1} \neq p_{\omega_{2}'+1} . $ Since $ d_{h}(p_{\omega_{1}'+1},p_{\omega_{2}'+1}) \geqslant 2t , $ it is easy to have
        $$ d_{h}(Enc_{\omega_{h}}(u_{1}),Enc_{\omega_{h}}(u_{2}))={d_{h}(u_{1},u_{2})+d_{h}(p_{\omega_{1}'+1},p_{\omega_{2}'+1})} \geqslant {2t+1}. $$
        Therefore, $ Enc_{\omega_{h}}:u \mapsto (u,p_{\omega_{h}(u)+1}) $ is an FCCHD.
    \end{proof}

        \begin{Corollary}
            For arbitrary integer $ k>2, $ $ r_{h}^{\omega_{h}}(k,1)=2 $ and $ r_{h}^{\omega_{h}}(k,2)=3. $
        \end{Corollary}

        \begin{proof}
            By Corollary \ref{Cor3}, $ r_{h}^{\omega_{h}}(k,1) \geqslant 2. $ As shown in Theorem \ref{Con1}, when $ t=1 $, there exists an FCCHD for homogeneous weight functions with redundancy $2$, thus $ r_{h}^{\omega_{h}}(k,1)=2. $

            By Corollary \ref{Cor3}, $ r_{h}^{\omega_{h}}(k,2) \geqslant 3. $ By Theorem \ref{Con1}, we know that $ r_{h}^{\omega_{h}}(k,2) \leqslant 3, $ thus $ r_{h}^{\omega_{h}}(k,2)=3. $
        \end{proof}

    \begin{Remark}
        When $ t=2 $ and $ t=3 $, redundancies of FCCHDs in Theorem \ref{Con1} reach the optimal redundancy bound of FCCHDs for homogeneous weight functions.
    \end{Remark}

    \subsection{FCCHDs for Homogeneous weight distribution functions}

    Let $ T \in \mathbb{N} $ with $ T|(2k+1) . $ We consider the homogeneous weight distribution function $ {f(u)=\Delta_{T_{h}}(u) \triangleq \lfloor{\frac{\omega_{h}(u)}{T}}\rfloor} $ on $ \mathbb{Z}_{2^{l}}^{k}. $ In this subsection, we let $ {E_{2} \triangleq |{\rm Im}(f)|=|{\rm Im}(\Delta_{T_{h}})|=\frac{2k+1}{T}} . $ Now, we provide an FCCHD for the function $ \Delta_{T_{h}} $ as follows.
    \begin{Theorem}\label{Con2}
        Assume that $ k,t,T \in \mathbb{N}, $ therein, $ T $ divides $ (2k+1) $ and $ T \geqslant {2t+1}. $ We define an encoding function $ Enc_{\Delta_{T_{h}}}:\mathbb{Z}_{2^{l}}^{k} \rightarrow \mathbb{Z}_{2^{l}}^{k+t}, $ $$ Enc_{\Delta_{T_{h}}}(u)=(u,p_{\omega_{h}(u)+1}), $$ with $ p_{i} \in \mathbb{Z}_{2^{l}}^{t} $ defined as follows. We set $ p_{i}=(\,\overset{i-1}{\overbrace{1,\,\cdots,\,1}}\,~\,\overset{t-i+1}{\overbrace{0,\,\cdots,\,0}}\,) $ for $ {\forall \,i \in [t+1]} ${\rm{;}} $ p_{i}=(\,\overset{i-1-t}{\overbrace{2^{l-1},\,\cdots,\,2^{l-1}}}\,~\,\overset{2t-i+1}{\overbrace{1,\,\cdots,\,1}}\,) $ for $ i \in \{t+2,\ldots,2t+1\} ${\rm{;}} when $ {i \in
        \{2t+2,\ldots,T\}} $, set $p_{i}=(\,\overset{t}{\overbrace{2^{l-1},\,\cdots,\,2^{l-1}}}\,)$ and $ p_{i}=p_{i \text{ smod } T} $ for $ i \geqslant {T+1} . $ Then this code is an {\rm {FCCHD}} for the Homogeneous weight distribution function $ \Delta_{T_{h}}. $ In addition, $ r_{h}^{\Delta_{T_{h}}}(k,t)=t $ also holds when $ T<2k+1 .$
    \end{Theorem}

    \begin{proof}
        Define encoding function as above, by Definition \ref{Def2}, we need to prove $$ d_{h}(Enc_{\Delta_{T_{h}}}(u_{1}),Enc_{\Delta_{T_{h}}}(u_{2})) \geqslant {2t+1}, $$ for all $ u_{1},u_{2} \in \mathbb{Z}_{2^{l}}^{k} $ with $ \Delta_{T_{h}}(u_{1}) \neq \Delta_{T_{h}}(u_{2}) . $

        If $ d_{h}(u_{1},u_{2}) \geqslant {2t+1} , $ it is easy to see that $$ d_{h}(Enc_{\Delta_{T_{h}}}(u_{1}),Enc_{\Delta_{T_{h}}}(u_{2}))=d_{h}(u_{1},u_{2})+d_{h}(p_{\omega_{h}(u_{1})+1},p_{\omega_{h}(u_{2})+1}) \geqslant {2t+1}.  $$
        If $ d_{h}(u_{1},u_{2})<{2t+1} $ and $ T \geqslant {2t+1} , $ we can assume that $ \Delta_{T_{h}}(u_{1})={m-1},~\Delta_{T_{h}}(u_{2})=m $ for some positive integer $m.$ Let $ \omega_{h}(u_{1})=(m-1)T+\omega_{1},~\omega_{h}(u_{2})=mT+\omega_{2} $ with $ \omega_{1},\omega_{2} \in \{0,1,\ldots,T-1\} , $ then $ p_{\omega_{h}(u_{1})+1}=p_{\omega_{1}+1},~p_{\omega_{h}(u_{2})+1}=p_{\omega_{2}+1} . $ We observe that $ \omega_{1}>\omega_{2} $ as $$ {\omega_{h}(u_{2})-\omega_{h}(u_{1})} \leqslant \omega_{h}(u_{2}-u_{1})<{2t+1} \leqslant T. $$
        It follows that
        \begin{equation}\label{equ}
           \begin{aligned}
               d_{h}(Enc_{\Delta_{T_{h}}}(u_{1}),Enc_{\Delta_{T_{h}}}(u_{2}))=&d_{h}(u_{1},u_{2})+d_{h}(p_{\omega_{h}(u_{1})+1},p_{\omega_{h}(u_{2})+1}) \\
               \geqslant &{T-\omega_{1}+\omega_{2}+d_{h}(p_{\omega_{1}+1},p_{\omega_{2}+1})}.
           \end{aligned}
        \end{equation}
        {\bf Case 1.}
        For $ T=2t+1 , $ then $ \omega_{1},\omega_{2} \leqslant 2t . $

        If $ \omega_{1},\omega_{2} \leqslant t $, we have $ p_{\omega_{1}+1}=(\,\overset{\omega_{1}}{\overbrace{1,\,\cdots,\,1}}\,~\,\overset{t-\omega_{1}}{\overbrace{0,\,\cdots,\,0}} \,),\, p_{\omega_{2}+1}=(\,\overset{\omega_{2}}{\overbrace{1,\,\cdots,\,1}}\,~\,\overset{t-\omega_{2}}{\overbrace{0,\,\cdots,\,0}}\,) , $ thus $$ d_{h}(p_{\omega_{h}(u_{1})+1},p_{\omega_{h}(u_{2})+1})={|\omega_{1}-\omega_{2}|} . $$ By Equation (\ref{equ}), $$ d_{h}(Enc_{\Delta_{T_{h}}}(u_{1}),Enc_{\Delta_{T_{h}}}(u_{2})) \geqslant {T-\omega_{1}+\omega_{2}+|\omega_{1}-\omega_{2}|} \geqslant {2t+1}.$$

        If $ t < \omega_{1},\omega_{2} \leqslant 2t , $ corresponding redundancies are $ p_{\omega_{1}+1}=(\,\overset{\omega_{1}-t}{\overbrace{2^{l-1},\,\cdots,\,2^{l-1}}}\,~\,\overset{2t-\omega_{1}}{\overbrace{1,\,\cdots,\,1}}\,) $ and ${p_{\omega_{2}+1}=(\,\overset{\omega_{2}-t}{\overbrace{2^{l-1},\,\cdots,\,2^{l-1}}}\,~\,\overset{2t-\omega_{2}}{\overbrace{1,\,\cdots,\,1}}\,) }, $ hence, $$ d_{h}(p_{\omega_{h}(u_{1})+1},p_{\omega_{h}(u_{2})+1})={|\omega_{1}-\omega_{2}|}. $$ Similarly, we can get $ d_{h}(Enc_{\Delta_{T_{h}}}(u_{1}),Enc_{\Delta_{T_{h}}}(u_{2}))\geqslant {2t+1} .$

        If $ \omega_{1}>t,~\omega_{2} \leqslant t , $ then $ p_{\omega_{1}+1}=(\,\overset{\omega_{1}-t}{\overbrace{2^{l-1},\,\cdots,\,2^{l-1}}}\,~\,\overset{2t-\omega_{1}}{\overbrace{1,\,\cdots,\,1}}\,),\,$ $ p_{\omega_{2}+1}=(\,\overset{\omega_{2}}{\overbrace{1,\,\cdots,\,1}}\,~\,\overset{t-\omega_{2}}{\overbrace{0,\,\cdots,\,0}}\,).$ If $ {\omega_{1}-t}>\omega_{2} , $ then $$ d_{h}(p_{\omega_{h}(u_{1})+1},p_{\omega_{h}(u_{2})+1})={2[(\omega_{1}-t)-\omega_{2}]+\omega_{2}+(2t-\omega_{1})}={\omega_{1}-\omega_{2}} .$$ If $ {\omega_{1}-t} \leqslant \omega_{2} , $ then $$ d_{h}(p_{\omega_{h}(u_{1})+1},p_{\omega_{h}(u_{2})+1})={\omega_{1}-\omega_{2}} .$$ Thus by Equation (\ref{equ}), $$ d_{h}(Enc_{\Delta_{T_{h}}}(u_{1}),Enc_{\Delta_{T_{h}}}(u_{2})) \geqslant {T-\omega_{1}+\omega_{2}+\omega_{1}-\omega_{2}}=T=2t+1. $$
        Hence, this encoding function is an FCCHD when $ T=2t+1. $

        {\bf Case 2.}
        For $ T>{2t+1} .$

        If $ \omega_{i} \leqslant\ 2t ,i=1,2 ,$ by the same argument for $ T=2t+1 , $ then we have $$ d_{h}(Enc_{\Delta_{T_{h}}}(u_{1}),Enc_{\Delta_{T_{h}}}(u_{2})) \geqslant {2t+1}. $$

        If $ \omega_{1}>2t,~t<\omega_{2} \leqslant2t , $ corresponding redundancies are $p_{\omega_{h}(u_{1})+1}=(\,\overset{t}{\overbrace{2^{l-1},\,\cdots,\,2^{l-1}}}\,),\,$ $ p_{\omega_{h}(u_{2})+1}=(\,\overset{\omega_{2}-t}{\overbrace{2^{l-1},\,\cdots,\,2^{l-1}}}\,~\,\overset{2t-\omega_{2}}{\overbrace{1,\,\cdots,\,1}}\,), $ hence, $$ d_{h}(p_{\omega_{h}(u_{1})+1},p_{\omega_{h}(u_{2})+1})={2t-\omega_{2}}. $$ Therefore, $$ d_{h}(Enc_{\Delta_{T_{h}}}(u_{1}),Enc_{\Delta_{T_{h}}}(u_{2})) \geqslant {T-\omega_{1}+\omega_{2}+2t-\omega_{2}}=T-\omega_{1}+2t \geqslant {2t+1}. $$

        If $ \omega_{1}>2t,~\omega_{2} \leqslant t , $ then $ p_{\omega_{h}(u_{1})+1}=(\,\overset{t}{\overbrace{2^{l-1},\,\cdots,\,2^{l-1}}}\,),~p_{\omega_{2}+1}=(\,\overset{\omega_{2}}{\overbrace{1,\,\cdots,\,1}}\,~\,\overset{t-\omega_{2}}{\overbrace{0,\,\cdots,\,0}}\,), $ thus $$ d_{h}(p_{\omega_{h}(u_{1})+1},p_{\omega_{h}(u_{2})+1})={2t-\omega_{2}}. $$ Hence,
        $$ d_{h}(Enc_{\Delta_{T_{h}}}(u_{1}),Enc_{\Delta_{T_{h}}}(u_{2})) \geqslant {T-\omega_{1}+\omega_{2}+2t-\omega_{2}} \geqslant {2t+1}. $$

        If $ \omega_{1}>2t $ and $ \omega_{2}>2t, $ then $ d_{h}(u_{1},u_{2}) \geqslant {T-\omega_{1}+\omega_{2}} \geqslant {2t+1} . $ obviously.

        Therefore, for $ T \geqslant {2t+1} $, this encoding function is an FCCHD with redundancy length of t. If $ T<2k+1, $ by Corollary \ref{Cor1}, $ r_{h}^{\Delta_{T_{h}}}(k,t) \geqslant t . $  Hence, $ r_{h}^{\Delta_{T_{h}}}(k,t)=t. $
    \end{proof}

    \subsection{FCCHDs for Rosenbloom-Tsfasnman weight functions}
    In this subsection, we introduce Rosenbloom-Tsfasnman weight functions and study the optimal redundancy of FCCHDs for Rosenbloom-Tsfasnman weight functions.
    \begin{Definition}\label{Def12}
        $ \forall u \in \mathbb{Z}_{2^{l}}^{k}, $ we suppose that $ u=(u_{1},u_{2},\ldots,u_{k}). $ We define the Rosenbloom-Tsfasnman weight function $ \omega_{RT}, $ if $ u \neq (0,0,\ldots,0), $ $$ \omega_{RT}(u)=\max{\{1 \leqslant j \leqslant k ~| ~ \forall s>j,u_{s}=0,u_{j} \neq 0 \}}; $$
        if $ u=(0,0,\ldots,0), $ we denote $ \omega_{RT}(u)=0. $
    \end{Definition}

    Refer to this definition, we have $ |{\rm Im}(\omega_{RT})|={k+1} $ and $ {\rm Im}(\omega_{RT})=\{0,1,\ldots,k\} $ on $ \mathbb{Z}_{2^{l}}^{k}. $
    \begin{Theorem}
        Let $f(u)=\omega_{RT}(u) $ and $ {\rm Im}(\omega_{RT}) \triangleq \{f_{0},f_{1},\ldots,f_{k}\} $ on $ \mathbb{Z}_{2^{l}}^{k}. $ Then entries of the $ (k+1) \times (k+1)$ matrix $D_{h}^{\omega_{RT}}(t,f_{0},f_{1},\ldots,f_{k})$ are
        \begin{equation*}
            [D_{h}^{\omega_{RT}}(t,f_{0},\ldots,f_{k})]_{(i+1)(j+1)}=\left\{
  	         \begin{aligned}
  	          &{2t},& \text{if } i \neq j , \\
                &0,& \text{otherwise},
               \end{aligned}
               \right.
        \end{equation*}
        for ${i,j \in \{\,0,1, \ldots, k\,\}} $. Furthermore, $ r_{h}^{\omega_{RT}}(k,t) \leqslant \left\lceil {\frac{\ln{(k+1)}+2t-1}{1-\ln{2}}} \right\rceil$.
    \end{Theorem}

    \begin{proof}
        Since $ {\rm Im}(\omega_{RT}) \triangleq \{f_{0},f_{1},\ldots,f_{k}\}=\{0,1,\cdots,k\} $ on $ \mathbb{Z}_{2^{l}}^{k}. $ Without loss of generality assume that $ f_{i}=i,\,{i \in \{\,0,1, \ldots, k\,\}} $. By Definition \ref{Def6},
        $$ d_{h}^{\omega_{RT}}(f_{i},f_{j})={ \min{ \{\, d_{h}(u,v) \, | \, \omega_{RT}(u)=f_{i} , \, \omega_{RT}(v)=f_{j},\, u,\,v \in \mathbb{Z}_{2^{l}}^{k}\, \} }}, $$

        If $i=j$, $d_{h}^{\omega_{RT}}(f_{i},f_{j})=0$. If $i \neq j$, suppose $ i>j $. Let $ u_{ij}=e_{i}+e_{j},\,v_{ij}=e_{j}$, where $ e_{i},\,e_{j} \in \mathbb{Z}_{2^{l}}^{k} $ are unit vectors with i-th, j-th components being 1, respectively. Thus $ \omega_{RT}(u_{ij})=i=f_{i} $ and $ \omega_{RT}(v_{ij})=j=f_{j} $. Since $ d_{h}(u_{ij},v_{ij})=1 $, we have $ d_{h}^{\omega_{RT}}(f_{i},f_{j})=1 $. Therefore, $ \forall \, i,\,j \in \{\,0,\,1,\,\cdots,\,k\,\} $,
        \begin{equation*}
            [D_{h}^{\omega_{RT}}(t,f_{0},\ldots,f_{k})]_{(i+1)(j+1)}=\left\{
  	         \begin{aligned}
  	          &{2t},& \text{if } i \neq j , \\
                &0,& \text{otherwise}.
               \end{aligned}
               \right.
        \end{equation*}

        Since $ N_{h}(D_{h}^{\omega_{RT}}(t,f_{0},\ldots,f_{k}))=N_{h}(k+1,\,2t) $, by Theorem \ref{Th2} and Theorem \ref{Th3}, we have $$ r_{h}^{\omega_{RT}}(k,t) \leqslant N_{h}(k+1,\,2t) \leqslant \left\lceil {\frac{\ln{(k+1)}+2t-1}{1-\ln{2}}} \right\rceil. $$
    \end{proof}

   \section{FCCHDs for two types of functions}
   In this section, we introduce two types of functions, which are locally homogeneous binary functions and min-max functions. Using the special structure of these functions, we further study FCCHDs for these functions.
    \subsection{FCCHDs for locally homogeneous binary functions}
	In this subsection, we introduce homogeneous binary functions and we study the optimal redundancy of FCCHDs for these functions.

    Firstly, we define a homogeneous function ball of a function $ f $ as follows.
	
	\begin{Definition}\label{Def8}
		The homogeneous function ball of a function f with radius $\rho$ around $u \in \mathbb{Z}_{2^{l}}^{k}$ is defined by
         $$ B_{h}^{f}(u,\rho)=\{f(u'):u' \in
         \mathbb{Z}_{2^{l}}^{k} \text{ and } d_{h}(u,u') \leqslant \rho \} . $$
	\end{Definition}

    Further, the locally homogeneous binary function is defined as follows.
    \begin{Definition}\label{Def9}
        A function $f:\mathbb{Z}_{2^{l}}^{k} \rightarrow {\rm Im}(f)$ is called $\rho$-locally homogeneous binary function, if for all $u \in \mathbb{Z}_{2^{l}}^{k},$ $|B_{h}^{f}(u,\rho)| \leqslant 2 .$
    \end{Definition}

    It is easy to prove the existence of this function. Constant functions and binary functions are $\rho$-locally homogeneous binary functions for any positive integer $\rho.$ Obviously, when $ T \geqslant 4t+1 , $ the homogeneous weight distribution function $ \Delta_{T_{h}} $ is a $ 2t $-locally homogeneous binary function. For any $2t$-locally homogeneous binary function, a more specific bound is obtained as follows.
    \begin{Theorem}\label{Lem7}
        For any $2t$-locally homogeneous binary function $f$ with $|{\rm Im}(f)| \geqslant 2,$ $$ t \leqslant r_{h}^{f}(k,t) \leqslant 2t. $$
    \end{Theorem}

    \begin{proof}
        By Corollary \ref{Cor1}, $r_{h}^{f}(k,t) \geqslant t.$
        On the other hand, let $ {\rm Im}(f)=\{f_{1},f_{2},\ldots,f_{E}\} $ and without loss of generality, assume that $f_{i}=i.$ Let $u$ be the message to be encoded. And define
        \begin{equation*}\label{Formula 4}
           w_{2t}(u)=\left\{
  	         \begin{aligned}
  	          &1,\text{ if } f(u)=\max{B_{h}^{f}(u,2t)}, \\
                &0,\text{ otherwise}.
               \end{aligned}
               \right.
        \end{equation*}
        Define $Enc(u)=(\,u,\,\overset{2t}{\overbrace{w_{2t}(u),\,\cdots,\,w_{2t}(u)}}\,)=(\,u,p\,),\,u \in \mathbb{Z}_{2^{l}}^{k} $, thus $p=(\,\overset{2t}{\overbrace{w_{2t}(u),\,\cdots,\,w_{2t}(u)}}\,).$ Suppose the receiver receives $(u',p').$ Firstly, we compute $B_{h}^{f}(u',t).$
        For any $ u_{1} \in \mathbb{Z}_{2^{l}}^{k} $ with $d_{h}(u',u_{1}) \leqslant t , $ we have
        $$ d_{h}(u,u_{1}) \leqslant {d_{h}(u',u_{1})+d_{h}(u',u)} \leqslant 2t, $$
        thus $f(u) \in B_{h}^{f}(u',t) \subseteq B_{h}^{f}(u,2t).$ Since $f$ is $2t$-locally homogeneous binary function, it holds that $|B_{h}^{f}(u,2t)|\leqslant 2$ and $ |B_{h}^{f}(u',t)|\leqslant 2. $

        If $ |B_{h}^{f}(u',t)|=1, $ we get the correct $f(u).$

        If $|B_{h}^{f}(u',t)|=2,$ since at most t errors occur through transmitting, we can compute $(w_{2t}(u'),p')$ to obtain the correct $w_{2t}(u).$ If $ w_{2t}(u)=1 , $ then $ f(u)=\max{B_{h}^{f}(u',t)} ; $ if $ w_{2t}(u)=0 , $ then $ f(u)=\min{B_{h}^{f}(u',t)}. $

        Therefore, we can get the correct value $ f(u) , $ and $ \{Enc(u)\,|\,u \in \mathbb{Z}_{2^{l}}^{k}\,\} $ is an FCCHD, thus $ t \leqslant r_{h}^{f}(k,t) \leqslant 2t. $
    \end{proof}

    The proof of Theorem \ref{Lem7} offers a way for encoding and decoding.

    \subsection{ FCCHDs for min-max functions}
    In this subsection, we introduce the definition of the min-max function and mainly study the optimal redundancy of FCCHDs for min-max functions. Assume $ k=ws $ for positive integers $ w,s. $ We write $ u \in \mathbb{Z}_{2^{l}}^{k} $ to be formed of $ w $ parts such that $ u=(u_{(1)},u_{(2)},\ldots,u_{(w)}) $ where $ u_{(i)} \in \mathbb{Z}_{2^{l}}^{s} $ of length $ s. $
    \begin{Definition}\label{Def11}{\rm (\cite{lbwy})}
        The min-max function on $ \mathbb{Z}_{2^{l}}^{k} $ is defined by
        $$ mm_{w}(u)=(\text{argmin}_{1 \leqslant i \leqslant w}{u_{(i)}},\text{argmax}_{1 \leqslant i \leqslant w}{u_{(i)}}), $$
        where $ u=(u_{(1)},u_{(2)},\ldots,u_{(w)}),~u_{(i)} \in \mathbb{Z}_{2^{l}}^{s},\,u \in \mathbb{Z}_{2^{l}}^{k} $ with $ k=ws $ and the ordering $ "<" $ between the $ u_{(i)} $'s is primarily lexicographical and secondarily, if $ u_{(i)}=u_{(j)}, $ according to ascending indices. Therein, $\text{argmin}_{1 \leqslant i \leqslant w}{u_{(i)}}$ is the position of the minimum $u_{(i)}$ in $u$ and $\text{argmax}_{1 \leqslant i \leqslant w}{u_{(i)}})$ is the position of the maximum $u_{(i)}$ in $u$.
    \end{Definition}

    \begin{Remark}
        Let $ \mathbb{Z}_{2^{l}}=\{~\Bar{0},~\Bar{1},~\ldots,~\overline{{2^{l}-1}}~\}. $ We first define the ordering $ "<" $ in $ \mathbb{Z}_{2^{l}} $ to be $ \Bar{0} < \Bar{1} < \cdots < \overline{{2^{l}-2}} < \overline{{2^{l}-1}}. $ Then the ordering $ "<" $ between the $ u^{(i)} $'s is primarily lexicographical, i.e., compare the ordering of their first components, if their first components are equal, then compare their second components, until the first unequal components, otherwise, they are equal.
    \end{Remark}

    Such as $ {u=(u_{(1)},u_{(2)},u_{(3)},u_{(4)},u_{(5)})=(00,01,10,00,02) \in \mathbb{Z}_{4}^{10}},$ the ordering between $ u_{(i)} $ is $ {u_{(1)}<u_{(4)}<u_{(2)}<u_{(5)}<u_{(3)}}, $ thus $ mm_{w}(u)=(1,3). $ By further analysis of the homogeneous distance between two min-max function values, we have the following results.
    \begin{Lemma}\label{Lem11}
         For $ w \geqslant 3,~s \geqslant 2, $ the minimum distance between any two min-max function values $ f_{1},f_{2} $ is at most $ 2 , $ i.e. $$ d_{h}^{mm_{w}}(f_{1},f_{2}) \leqslant 2,~{\forall \ f_{1},f_{2} \in {\rm Im}(mm_{w})}. $$
    \end{Lemma}

    \begin{proof}
        We give the proof for $ s=2. $ For $ s>2 , $ we can restrict all the bits of all $ u_{(v)},v \in [w] $ to be $ 0 $ except for the two least significant bits and apply the same proof method of $ s=2 . $ By Definition \ref{Def6}, we need to show that $ {\forall \ i,j,i',j' \in [w]}  \mbox{~with~}  (i,j) \neq (i',j')$, $$ {\exists \ u,u' \in \mathbb{Z}_{2^{l}}^{k} \mbox{~with~} mm_{w}(u)=(i,j),~mm_{w}(u')=(i',j')} \mbox{~s.t.~} d_{h}(u,u') \leqslant 2. $$

        If $ i' \leqslant i , $ let $$ u=(01,\ldots,\overset{u_{(i)}}{\overbrace{00}},\ldots,\overset{u_{(j)}}{\overbrace{02^{l-1}}},01,\ldots,01),$$ then, $mm_{w}(u)=(i,j). $
        Firstly, if $ i'<i , $ flip the second bit of $ u_{(i')} $ to be $ 0 $ such that $ u'_{(i')}=(00) , $ then change the first bit of $ u_{(j')} $ to be $ 1 $ such that $ u'_{(j')}=(11) $ or $ u'_{(j')}=(12^{l-1}) , $ thus we obtain $$ u'=(01,\ldots,\overset{u'_{(i')}}{\overbrace{00}},01,\ldots,\overset{u'_{(i)}}{\overbrace{00}},\ldots,\overset{u'_{(j')}}{\overbrace{11}},\ldots,\overset{u'_{(j)}}{\overbrace{02^{l-1}}},01,\ldots,01), $$
        or $$ u'=(01,\ldots,\overset{u'_{(i')}}{\overbrace{00}},01,\ldots,\overset{u'_{(i)}}{\overbrace{00}},\ldots,01,\ldots,\overset{u'_{(j')}}{\overbrace{12^{l-1}}},01,\ldots,01), $$ we see that $ mm_{w}(u')=(i',j') $ and $ d_{h}(u,u')=2 , $ therefore $ d_{h}^{mm_{w}}((i,j),(i',j')) \leqslant 2 . $

        If $ i'>i,~i' \neq j , $ let
        $$  u=(01,\ldots,\overset{u_{(i)}}{\overbrace{00}},\ldots,01,\overset{u_{(i')}}{\overbrace{00}},01,\ldots,\overset{u_{(j)}}{\overbrace{02^{l-1}}},01,\ldots,01), $$ then $ mm_{w}(u)=(i,j). $
        Firstly, we flip the second bit of $ u_{(i)} $ to be $ 1 $ such that $ u'_{(i)}=(01) , $ then change the first bit of $ u_{(j')} $ to be $ 1 $ such that $ u'_{(j')}=(11) $ or $ u'_{(j')}=(12^{l-1}) , $ thus we have $$ u'=(01,\ldots,01,\overset{u'_{(i)}}{\overbrace{01}},\ldots,\overset{u'_{(i')}}{\overbrace{00}},\ldots,\overset{u'_{(j')}}{\overbrace{11}},\ldots,\overset{u'_{(j)}}{\overbrace{02^{l-1}}},01,\ldots,01), $$
        or
        $$ u'=(01,\ldots,01,\overset{u'_{(i)}}{\overbrace{01}},\ldots,\overset{u'_{(i')}}{\overbrace{00}},\ldots,\overset{u'_{(j')}}{\overbrace{12^{l-1}}},01,\ldots,01), $$
        we know $ mm_{w}(u')=(i',j') $ and $ d_{h}(u,u')=2 , $ thus $ d_{h}^{mm_{w}}((i,j),(i',j')) \leqslant 2 . $

        If $ i'>i,~i'=j , $ let
        $$ u=(02^{l-1},\ldots,\overset{u_{(i)}}{\overbrace{01}},\ldots,\overset{u_{(j)}=u_{(i')}}{\overbrace{10}},02^{l-1},\ldots,02^{l-1}), $$ then $mm_{w}(u)=(i,j). $
        Firstly, we change the first bit of $ u_{(j)} $ to be $ 0 $ such that $ u'_{(j)}=(00) , $ then change the first bit of $ u_{(j')} $ to be 1 such that $ u'_{(j')}=(11) $ or $ u'_{(j')}=(12^{l-1}) ,$ thus we obtain
        $$ u'=(02^{l-1},\ldots,\overset{u'_{(i)}}{\overbrace{01}},\ldots,\overset{u'_{(j')}}{\overbrace{12^{l-1}}},\ldots ,\overset{u'_{(j)}=u'_{(i')}}{\overbrace{00}},02^{l-1},\ldots,02^{l-1}), $$
        or
        $$ u'=(02^{l-1},\ldots,\overset{u'_{(i)}=u'_{(j')}}{\overbrace{11}},\ldots,\overset{u'_{(j)}=u'_{(i')}}{\overbrace{00}},02^{l-1},\ldots,02^{l-1}). $$
        Then, $ mm_{w}(u)=(i',j') $ and $ d_{h}(u,u')=2 , $ therefore, $ d_{h}^{mm_{w}}((i,j),(i',j')) \leqslant 2 . $
    \end{proof}

    \begin{Lemma}\label{Lem12}
        For $ w \geqslant 3,~s \geqslant 2,~l \geqslant 2 , $ given a min-max function value $ f_{1}=(i,j) , $ the number of min-max function values $ f_{2} \neq (i,j) $ that satisfy $ d_{h}^{mm_{w}}(f_{1},f_{2})=1 $ is $ 4(w-2) . $ Therefore, the number of entries in $ D_{h}^{mm_{w}}(t,f_{1},f_{2},\ldots,f_{w(w-1)}) $ that equal $ 2t $ is equal $ 4w(w-1)(w-2) . $
    \end{Lemma}

    \begin{proof}
        For $ s>2 , $ we can restrict all bits of all $ u_{(v)},v \in [w] $ to be $ 0 $ except for the two least significant bits and apply the proof method of $ s=2 . $ Thus we just need to prove the case of $ s=2 . $

        {\bf Case 1.} Change one bit of $ u_{(i)} . $

        (1) Firstly, change $ u_{(i)} $ such that the result becomes larger than $ u_{(i)} $ but smaller than $ u_{(j)} . $ Let
        $$ u=(11,\ldots,11,\overset{u_{(i)}}{\overbrace{01}},11,\ldots,11,\overset{u_{(v)}}{\overbrace{10}},11,\ldots,11,\overset{u_{(j)}}{\overbrace{12^{l-1}}},11,\ldots,11),~mm_{w}(u)=(i,j). $$
        Change the first bit of $ u_{(i)} $ to be $ 1 $ that is $ u'_{(i)}=(11) $ to get $ u' , $
        $$ u'=(11,\ldots,11,\overset{u'_{(i)}}{\overbrace{11}},11,\ldots,11,\overset{u'_{(v)}}{\overbrace{10}},11,\ldots,11,\overset{u'_{(j)}}{\overbrace{12^{l-1}}},11,\ldots,11), $$
        then $ d_{h}(u,u')=1,~mm_{w}(u')=(v,j),~v \in [w]\backslash \{i,j\} . $ Thus $ f_{2}=(v,j),~v \in [w]\backslash \{i,j\} $ satisfies $ d_{h}^{mm_{w}}(f_{1},f_{2})=1. $

        (2) Change $ u_{(i)} $ such that the result becomes larger than $ u_{(j)} . $ Let
        $$ u=(10,\ldots,10,\overset{u_{(i)}}{\overbrace{02}},10,\ldots,10,\overset{u_{(v)}}{\overbrace{03}},10,\ldots,\overset{u_{(j)}}{\overbrace{10}},11,10,\ldots,10),~mm_{w}(u)=(i,j). $$
        Change the first bit of $ u_{(i)} $ to be $ 1 $ that is $ u'_{(i)}=(12) $ to get $ u' , $
        $$ u'=(10,\ldots,10,\overset{u'_{(i)}}{\overbrace{12}},10,\ldots,10,\overset{u'_{(v)}}{\overbrace{03}},10,\ldots,10,\overset{u'_{(j)}}{\overbrace{11}},10,\ldots,10), $$
        then $ d_{h}(u,u')=1,~mm_{w}(u')=(v,i),~v \in [w]\backslash \{i,j\} . $ Thus $ f_{2}=(v,i),~v \in [w]\backslash \{i,j\} $ satisfies $ d_{h}^{mm_{w}}(f_{1},f_{2})=1. $

        {\bf Case 2.} Change one bit of $ u_{(j)} . $\

        (1) Firstly, change $ u_{(j)} $ such that the result becomes smaller than $ u_{(j)} $ but larger than $ u_{(i)} . $ Let
        $$ u=(10,\ldots,10,\overset{u_{(i)}}{\overbrace{00}},10,\ldots,10,\overset{u_{(v)}}{\overbrace{11}},10,\ldots,10,\overset{u_{(j)}}{\overbrace{20}},10,\ldots,10),~mm_{w}(u)=(i,j). $$
        Change the first bit of $ u_{(j)} $ to be $ 1 $ that is $ u'_{(j)}=(10) $ to get $ u' , $
        $$ u'=(10,\ldots,10,\overset{u'_{(i)}}{\overbrace{00}},10,\ldots,10,\overset{u'_{(v)}}{\overbrace{11}},10,\ldots,10,\overset{u'_{(j)}}{\overbrace{10}},10,\ldots,10), $$
        then $ d_{h}(u,u')=1,~mm_{w}(u')=(i,v),~v \in [w]\backslash \{i,j\} . $ Thus $ f_{2}=(i,v),~v \in [w]\backslash \{i,j\} $ satisfies $ d_{h}^{mm_{w}}(f_{1},f_{2})=1. $

        (2) Change $ u_{(j)} $ such that the result becomes smaller than $ u_{(i)} . $ Let
        $$ u=(12,\ldots,12,\overset{u_{(i)}}{\overbrace{11}},12,\ldots,12,\overset{u_{(v)}}{\overbrace{13}},12,\ldots,12,\overset{u_{(j)}}{\overbrace{20}},12,\ldots,12),~mm_{w}(u)=(i,j). $$
        Change the first bit of $ u_{(j)} $ to be $ 1 $ that is $ u'_{(j)}=(10) $ to get $ u' , $
        $$ u'=(12,\ldots,12,\overset{u'_{(i)}}{\overbrace{11}},12,\ldots,12,\overset{u'_{(v)}}{\overbrace{13}},12,\ldots,12,\overset{u'_{(j)}}{\overbrace{10}},12,\ldots,12), $$
        then $ d_{h}(u,u')=1,~mm_{w}(u')=(j,v),~v \in [w]\backslash \{i,j\} . $ Thus $ f_{2}=(j,v),~v \in [w]\backslash \{i,j\} $ satisfies $ d_{h}^{mm_{w}}(f_{1},f_{2})=1. $

        {\bf Case 3.} Change one bit of $ u_{(v)},~v \in [w]\backslash \{i,j\} . $

        Obviously, this does not yield any new function values that satisfy the condition.

        To sum up, for each fixed min-max function value $ f_{1} , $ there exists $ 4(w-2) $ min-max function values $ f_{2} $ with $ f_{1} \neq f_{2} $ and $ d_{h}^{mm_{w}}(f_{1},f_{2})=1. $ Moreover, by Definition \ref{Def7}, the number of entries in $ D_{h}^{mm_{w}}(t,f_{1},f_{2},\ldots,f_{w(w-1)}) $ that equal $ 2t $ is equal $ {4w(w-1)(w-2)}. $
    \end{proof}

    Based on Lemma \ref{Lem11} and Lemma \ref{Lem12}, we obtain an upper bound of the optimal redundancy of FCCHDs for the min-max function.
    \begin{Theorem}\label{Th6}
        For $ w \geqslant3,s \geqslant 2,l \geqslant 2, $ the optimal redundancy $ r_{h}^{mm_{w}}(k,t) $ of {\rm{FCCHDs}} for the min-max function is bounded by
        $$ \frac{2t(w^{2}-w-1)-(3w^{2}-7w+5)}{w(w-1)} \leqslant r_{h}^{mm_{w}}(k,t) \leqslant \min_{r \in N}{\{r:\phi(r)>0\}}, $$
        where $ {\phi(r)={2^{lr}-(w^{2}-5w+7)V_{h}(r,2t-2)-4(w-2)V_{h}(r,2t-1) }}. $ Particularly, if $ {t=1},$ then $ r_{h}^{mm_{w}}(k,t) \geqslant \frac{2(2t-1)(w-2)}{w(w-1)}. $
    \end{Theorem}

    \begin{proof}
        Firstly, we prove the left side of the inequation.

        Let $ u_{i,j} \in \mathbb{Z}_{2^{l}}^{k} $ be $ w(w-1) $ information vectors, where $ i,j \in [w],~i \neq j $ and let
        $ D_{h}^{mm_{w}}(t,u_{1,2},u_{1,3},\ldots,u_{w-1,w}) $ be a homogeneous distance requirement matrix. By Corollary \ref{Cor1} and Lemma \ref{Lem3}, we have
       \begin{equation*}
           \begin{aligned}
           r_{h}^{mm_{w}}(k,t) \geqslant& N_{h}(D_{h}^{mm_{w}}(t,u_{1,2},u_{1,3},\ldots,u_{w-1,w})) \\
           \geqslant& \frac{2}{w^{2}(w-1)^{2}}\sum_{i,j:i<j}{[D_{h}^{mm_{w}}(t,u_{1,2},u_{1,3},\ldots,u_{w-1,w})]_{ij}}.
       \end{aligned}
       \end{equation*}
       The proof of arbitrary $ s>2 $ follows the same steps after setting the $ {s-2} $ left-most bits in every part of $ u_{i,j} $ to be $ 0 . $ Thus we just need to prove the case that $ s=2 .$

       For $ s=2, $ choose the representative information vectors $u_{i,j}$ to be
       $$ u_{i,j}=(01,\ldots,01,\overset{u_{(i)}}{\overbrace{00}},01,\ldots,\overset{u_{(j)}}{\overbrace{11}},01,\ldots,01), $$
       where $ i,j \in [w] $ and $ i \neq j . $ Note that $ mm_{w}(u_{i,j})=(i,j) $ and thus the corresponding function values are distinct. Further, let $ i',j' \in [w] $ with $ i \neq i',~j \neq j' . $ Then
       $ d_{h}(u_{i,j},u_{i',j})=d_{h}(u_{i,j},u_{i,j'})=2,~d_{h}(u_{i,j},u_{i',j'})=4 . $
       Thus for given $ u_{i,j} , $ there are $ {w-2} $ information vectors $ u_{i,j'} $ satisfying that the homogeneous distance to $ u_{i,j} $ equal $2,$ $ {w-2} $ information vectors $ u_{i',j} $ satisfying $ d_{h}(u_{i',j},u_{i,j})=2, $ and there are $ (w-1)(w-2)+1 $ $ u_{i',j'} $ such that $ d_{h}(u_{i,j},u_{i',j'})=4 . $

       Then, for $ t \geqslant 2 , $ each row of $ D_{h}^{mm_{w}}(t,u_{1,2},u_{1,3},\ldots,u_{w-1,w}) $ has $ 2(w-2) $ entries equal to $ 2t-1 ,$ $ (w-1)(w-2)+1 $ entries equal to $ 2t-3 .$ Thus
       \begin{equation*}
           \begin{aligned}
           r_{h}^{mm_{w}}(k,t)
           \geqslant& \frac{2}{w^{2}(w-1)^{2}}\sum_{i,j:i<j}{[D_{h}^{mm_{w}}(t,u_{1,2},u_{1,3},\ldots,u_{w-1,w})]_{ij}}\\
           =&\frac{2t(w^{2}-w-1)-(3w^{2}-7w+5)}{w(w-1)}.
       \end{aligned}
       \end{equation*}
       For $ t=1 , $ each row of $ D_{h}^{mm_{w}}(t,u_{1,2},u_{1,3},\ldots,u_{w-1,w}) $ has $ 2(w-2) $ entries equal to $ 2t-1 , $ other entries are $ 0 . $ Thus, $ r_{h}^{mm_{w}}(k,t) \geqslant \frac{2(2t-1)(w-2)}{w(w-1)}. $

       Then, we just need to show the right side.

        By Lemma \ref{Lem4}, denote $ D_{h}^{mm_{w}} \triangleq D_{h}^{mm_{w}}(t,f_{1},\ldots,f_{w(w-1)}) , $ by using Theorem \ref{Th2}, we have $$ r_{h}^{mm_{w}}(k,t) \leqslant \min_{r \in \mathbb{N}}{ \left \{r:2^{lr} > \max_{i \in [w(w-1)]}{\sum_{j=1}^{i-1}{V_{h}(r,{[D_{h}^{mm_{w}}]_{\pi(i)\pi(j)}-1})}} \right \}}. $$
        Let $$ \psi(r)={2^{lr}-\max_{i \in [w(w-1)]}{\sum_{j=1}^{i-1}{V_{h}(r,{[D_{h}^{mm_{w}}]_{\pi(i)\pi(j)}-1})}}}, $$ then $ r_{h}^{mm_{w}}(k,t) \leqslant \min_{r \in \mathbb{N}}{\{r:\psi(r)>0\}} , $ where $ \pi $ is a permutation of integers in $ {[w(w-1)]} . $ The maximum of this summing can be bounded from above by setting $ {i=w(w-1)} $ and choosing a row with the largest number of nonzero entries.

        By Lemma \ref{Lem11} and Lemma \ref{Lem12}, a row $ \pi(i) $ with maximum entries contains one entry to be $ 0 ,$ $ {4(w-2)} $ entries to be $ 2t $ and at most $ {w(w-1)-4(w-2)-1} $ entries equal $ {2t-1} . $ Thus let
        $$ {\phi(r)={2^{lr}-(w^{2}-5w+7)V_{h}(r,2t-2)-4(w-2)V_{h}(r,2t-1)}}, $$
        we have $ r_{h}^{mm_{w}}(k,t) \leqslant \min_{r \in \mathbb{N}}{\{r:\phi(r)>0\}} . $
    \end{proof}
    Based on the proof of Theorem \ref{Th6}, we obtain the following corollary by using the connection between homogeneous distance balls and Hamming distance balls.
    \begin{Corollary}\label{Cor4}
        For $ w \geqslant 3,\,s \geqslant 2 , $ the optimal redundancy $ r_{h}^{mm_{w}}(k,t) $ satisfies
        $$ r_{h}^{mm_{w}}(k,t) \geqslant {\frac{1}{l}(\log(w(w-1))+{\lfloor{\frac{t-2}{2}}\rfloor}\log{[\frac{1}{l}\log(w(w-1))]}-{\lfloor{\frac{t-2}{2}}\rfloor}\log{{\lfloor{\frac{t-2}{2}}\rfloor}})},\,t \geqslant 2. $$
    \end{Corollary}

    \begin{proof}
        Due to Corollary \ref{Cor1}, we know that
        $$ r_{h}^{mm_{w}}(k,t) \geqslant N_{h}(D_{h}^{mm_{w}}(t,u_{1,2},u_{1,3},\ldots,u_{w-1,w})) . $$ Refer to the proof of Theorem \ref{Th6}, we have $ d_{h}(u_{i,j},u_{i',j'}) \leqslant 4 $ for any $ i,j,i',j' , $ then $ N_{h}(D_{h}^{mm_{w}}(t,u_{1,2},u_{1,3},\ldots,u_{w-1,w})) \geqslant N_{h}(w(w-1),2t-3) ,$ when $ t \geqslant 2 . $

        Let $ r \triangleq N_{h}(w(w-1),2t-3) , $ using the sphere packing bound, it holds that $$ {(2^{l})^{r} \geqslant w(w-1)V_{h}(r,t-2)}, $$  where $ V_{h}(r,t-2)={|\{\,a \in \mathbb{Z}_{2^{l}}^{r}\,|\,d_{h}(a,u) \leqslant {t-2}\,\}|}. $ Since
        $$ \{\,a \in \mathbb{Z}_{2^{l}}^{r}\,|\,d_{h}(a,u) \leqslant {t-2}\} \supseteq \{\,a \in \mathbb{Z}_{2^{l}}^{r}\,|\,d(a,u) \leqslant {\lfloor{\frac{t-2}{2}}\rfloor} \,\}, $$
        we have $ V_{h}(r,t-2) \geqslant V_{2^{l}}(r,{\lfloor{\frac{t-2}{2}}\rfloor}), $ where $ V_{2^{l}}(r,{\lfloor{\frac{t-2}{2}}\rfloor})={|\{\,a \in \mathbb{Z}_{2^{l}}^{r}\,|\,d(a,u) \leqslant {\lfloor{\frac{t-2}{2}}\rfloor}\,\}|}. $ Thus $$ (2^{l})^{r} \geqslant w(w-1)V_{2^{l}}(r,{\lfloor{\frac{t-2}{2}}\rfloor}) . $$
        Using logarithmic operation, we have $$ lr \geqslant {\log{w(w-1)}+\log{V_{2^{l}}(r,{\lfloor{\frac{t-2}{2}}\rfloor})}}. $$ Note that $ V_{2^{l}}(r,{\lfloor{\frac{t-2}{2}}\rfloor})=\sum_{j=0}^{{\lfloor{\frac{t-2}{2}}\rfloor}}{\left(\begin{array}{c}
             r  \\
             j
        \end{array} \right)(2^{l}-1)^j} \geqslant \sum_{j=0}^{{\lfloor{\frac{t-2}{2}}\rfloor}}{\left( \begin{array}{c}
        r\\j \end{array} \right)}, $ then $$ lr \geqslant {log{w(w-1)}+\log{\sum_{i=0}^{{\lfloor{\frac{t-2}{2}}\rfloor}}\left( \begin{array}{c}
        r\\j \end{array} \right)}}. $$ Since $ \frac{b}{a} \geqslant \frac{b+1}{a+1} $ when $ 0 < a \leqslant b , $ we have $$ lr \geqslant {\log{w(w-1)}+{\lfloor{\frac{t-2}{2}}\rfloor\log{\frac{r}{{\lfloor{\frac{t-2}{2}}\rfloor}}}}}. $$ Thus, $$ lr \geqslant {\log(w(w-1))+{\lfloor{\frac{t-2}{2}}\rfloor}\log{[\frac{1}{l}\log(w(w-1))]}-{\lfloor{\frac{t-2}{2}}\rfloor}\log{{\lfloor{\frac{t-2}{2}}\rfloor}}}.$$
       Therefore,
       $$ r_{h}^{mm_{w}}(k,t) \geqslant {\frac{1}{l}\{\log(w(w-1))+{\lfloor{\frac{t-2}{2}}\rfloor}\log{[\frac{1}{l}\log(w(w-1))]}-{\lfloor{\frac{t-2}{2}}\rfloor}\log{{\lfloor{\frac{t-2}{2}}\rfloor}}\}},\,t \geqslant 2. $$
    \end{proof}

  \section{Conclusion}
     In this paper, we introduce FCCHDs on the residue ring $ \mathbb{Z}_{2^{l}}. $ This encoding method reduces the redundancy compared to classical error-correcting codes as receiver only need to get correct function values of messages which are the attribute of messages. We obtain some bounds of the optimal redundancy of FCCHDs for an arbitrary function and some specific functions. In order to obtain these bounds, we define $ D $-homogeneous distance codes and two classes of matrices, which are homogeneous distance requirement matrices and function homogeneous distance matrices. For some specific matrices $ D $, we find connections between the smallest length of these codes and the optimal redundancy of FCCHDs by constructing $ D $-homogeneous distance codes.

     Based on the structure of five classes of functions, we further study the optimal redundancy of FCCHDs for these functions and we obtained more accurate bounds of the optimal redundancy. We provide some constructions of FCCHDs for homogeneous weight functions and homogeneous weight distribution functions. Specially, redundancies of some of these codes reach the optimal redundancy bounds of FCCHDs. Since it's difficult to select information vectors to construct a homogeneous distance requirement matrix to  be equivalent to the function homogeneous distance matrix, we don't provide constructions of FCCHDs for Rosenbloom-Tsfasnman weight functions, locally homogeneous binary functions and min-max functions reaching the optimal redundancy bounds of FCCHDs. Next, we can try to find some special codes to construct FCCHDs or find some special matrix to construct $ D $-homogeneous distance codes to improve the redundancy bounds of FCCHDs for further study.

\bigskip 

{\bf  Acknowledgement}. This work was supported by NSFC (Grant Nos. 12441102 and 12271199).

\end{document}